\documentclass[11pt,a4paper]{article}
\usepackage[margin=1in]{geometry}
\usepackage{amsmath,amssymb,amsthm,mathtools}
\usepackage{graphicx,booktabs,array}
\usepackage{xcolor}
\usepackage[colorlinks=true,linkcolor=blue!60!black,citecolor=blue!60!black,urlcolor=blue!60!black]{hyperref}
\usepackage{caption}
\usepackage{placeins}
\newtheorem{theorem}{Theorem}
\newtheorem{proposition}[theorem]{Proposition}
\newtheorem{lemma}[theorem]{Lemma}
\newtheorem{corollary}[theorem]{Corollary}
\theoremstyle{definition}
\newtheorem{definition}[theorem]{Definition}
\newtheorem{remark}[theorem]{Remark}

\newcommand{\F}{\mathbb{F}}
\newcommand{\FG}{\F_2[G]}
\newcommand{\supp}{\operatorname{supp}}
\newcommand{\rowsp}{\operatorname{rowsp}}
\newcommand{\reg}{\operatorname{reg}}
\newcommand{\ZZ}{\mathbb{Z}}
\newcommand{\Lm}{\mathsf{L}}
\newcommand{\Rm}{\mathsf{R}}
\newcommand{\Cz}{C_G}
\newcommand{\Code}{\mathcal{C}}
\newcommand{\Aut}{\operatorname{Aut}}

\newcommand{\Out}{\operatorname{Out}}
\newcommand{\Hol}{\operatorname{Hol}}
\newcommand{\Tw}{\operatorname{Tw}}
\newcommand{\Tr}{\operatorname{Tr}}
\newcommand{\Aba}{\Aut_{\mathrm{ba}}}
\DeclareMathOperator{\rank}{rank}
\DeclareMathOperator{\Frob}{Frob}

\title{Lifted surgery for non-Abelian two-block group-algebra codes}
\author{\href{https://orcid.org/0000-0001-7448-5723}{Tushar Pandey}}
\date{}

\begin{document}
\maketitle

\begin{abstract}
Code surgery measures a logical operator of a quantum LDPC code with $O(d)$ rounds of syndrome
extraction. Lifted surgery amortises this cost on Abelian group-algebra codes by measuring an
orbit of logical operators under a translation symmetry in one merged code. We extend it to
two-block group-algebra codes over non-Abelian groups and ask whether non-commutativity lets one
merged code measure more operators than any Abelian group of its symmetries. We
classify the block-affine automorphisms of these codes and find that, measured against this full
group, most apparent non-Abelian gains disappear. We prove that the gain is at most the index of
a largest Abelian subgroup. We find codes with exact distance up to $13$ where the gain is
two, including codes in which one merged code reads out every logical qubit, and codes
over products of $A_4$, $S_4$ and $\mathrm{SL}(2,3)$ with gain three for a best logical, up to
$[[336,26,12]]$. A merged-distance lemma gives a simple condition for the gadget to preserve
the code distance. In circuit-level simulations with Relay-BP decoding, at the error
rates we can resolve, the non-Abelian gadget is as reliable as, or more reliable than, the
Abelian gadgets it replaces within statistical error, while using two to three times fewer
rounds.
\end{abstract}

\section{Introduction}\label{sec:intro}

Quantum low-density parity-check (qLDPC) codes achieve constant encoding rate with stabiliser
checks of bounded weight \cite{panteleev2022,bravyi2024}. In code surgery, the qLDPC analogue of
lattice surgery \cite{horsman2012}
\cite{cohen2022,cowtan2024css,cross2024,williamson2024,ide2024homological}, a logical Pauli
operator is measured by attaching an ancilla system to its support, running $O(d)$ rounds of
syndrome extraction in the resulting merged code, and detaching. Recent work has reduced the
space \cite{cowtan2025parallel,he2025extractors,swaroop2025adapters} and time
\cite{cowtan2025fast,chang2026constant} overhead of individual measurements and has shown how to
measure several operators in parallel \cite{cowtan2025parallel}. When many operators must be
measured, as in Pauli-based computation on a block with $k$ logical qubits, the number of
syndrome rounds remains a bottleneck.

Lifted surgery \cite{berent2026lifted} addresses this for Abelian group-algebra codes such as
bivariate-bicycle and radial codes. A translation symmetry maps a logical operator to an orbit of
logical operators, and an orbit with pairwise disjoint supports can be measured in one merged
code. On quantum radial codes this reduces the number of rounds by an order of magnitude.

Several of the most efficient small qLDPC codes are built from non-Abelian groups. Two-block
group-algebra (2BGA) codes, the non-Abelian generalisation of bivariate-bicycle codes, were
introduced by Lin and Pryadko and by Wang, Lin and Pryadko \cite{lin2024quantum,wang2023abelian}.
Mitten codes \cite{bhardwaj2026mitten}, lifted-product codes over non-Abelian groups with a
$1\times2$ base matrix, evade distance bounds constraining Abelian two-block codes, and their
parallel surgery measures many canonical logicals in one merged code using the orbit of one
representative under multiplication by the group. That action is not a code automorphism.
Canonical bases are also known for cyclic lifted-product codes \cite{zheng2026canonical}. To our knowledge it has not been examined
whether the automorphism symmetry on which lifted surgery relies persists when the group is
non-commutative, or whether non-commutativity increases the number of operators that one merged
code can measure.

We study both questions for 2BGA codes. A conjugation-invariant check element lets the whole
non-Abelian group act by left translations at check weight six, and against Abelian groups of
left and right translations this action appears to give gains of up to two. But 2BGA codes have further automorphisms,
among them right translations by different elements on the two blocks, maps induced by outer
automorphisms of the group, block exchanges and relabellings by central units. Abelian groups
built from these restore the parallelism in most cases, and the translation group is not even an
invariant of the code. Our comparisons therefore use the full block-affine automorphism group on
both sides.

We classify the Tanner-level block-affine automorphisms of 2BGA codes (Theorem~\ref{thm:class})
and give a criterion for a code to be rigid, meaning that its only block-affine automorphisms are
the translations (Proposition~\ref{prop:rigid}). Against the full block-affine group a gain of two
remains: we certify it exactly in ten rigid codes with distances from six to thirteen and in one
non-rigid code, and in four of them one merged code reads out every logical qubit
(Table~\ref{tab:family}). A coset-cover bound limits the gain by the index of an Abelian subgroup
(Proposition~\ref{thm:cover}), and guided by it we certify $16$ rigid codes with gain at least
three for a best logical, $13$ of them also at code level, with distance up to twelve
(Table~\ref{tab:gain3}).

A merged-distance lemma shows that the parallel gadget keeps both code distances whenever every
ancilla graph has Cheeger constant at least one (Lemma~\ref{lem:merged}). We simulate the gadgets
under circuit-level noise, with CNOT schedules chosen against hook errors and Relay-BP decoding of
the $Z$ sector \cite{muller2025relay}. At the error rates we can resolve, from $p=2\times10^{-3}$
to $5\times10^{-3}$ ($3$ and $4\times10^{-3}$ for $[[240,18,10]]$), the full gadget on each of
the four non-Abelian codes simulated is more reliable than the Abelian gadgets it replaces or
indistinguishable from them, and it uses two to three times fewer rounds. These error rates are
close to the pseudo-thresholds of the codes, and lower ones remain to be tested. All
constructions are available in the open-source package \texttt{liftsurg}.

Section~\ref{sec:prelim} reviews 2BGA codes and gauging surgery, Sections~\ref{sec:aut}
to~\ref{sec:codes} develop the automorphism theory, the gain and the certified codes,
Section~\ref{sec:ft} treats fault tolerance and Section~\ref{sec:discussion} discusses the
results. Proofs, further results and methods are in the appendices.

\section{Preliminaries}\label{sec:prelim}

\subsection{Two-block group-algebra codes}\label{sec:2bga}

For a CSS code with check matrices $H_X$ and $H_Z$ we write $L_Z=\ker H_X/\rowsp H_Z$,
$L_X=\ker H_Z/\rowsp H_X$ and $k=n-\rank H_X-\rank H_Z$. A qubit permutation is an automorphism if
it preserves both row spaces.

Let $G$ be a finite group and $\FG$ its group algebra over $\F_2$, the formal sums of group
elements with coefficients $0$ or $1$, multiplied by the group law. An element $a\in\FG$ is
identified with its support $\supp(a)\subseteq G$, and $a^*=\sum_ga_gg^{-1}$. Let $\Lm(a)$ and
$\Rm(b)$ be the matrices of left and right multiplication, $x\mapsto ax$ and $x\mapsto xb$, in
the basis $G$. Left and right multiplication commute for any group, so
\begin{equation}\label{eq:2bga}
H_X=[\Lm(a)\;|\;\Rm(b)],\qquad H_Z=[\Rm(b)^{\mathsf T}\;|\;\Lm(a)^{\mathsf T}]
\end{equation}
defines a CSS code $\Code(a,b)$ on $2|G|$ qubits with check weight $|\supp a|+|\supp b|$, the
two-block group-algebra code \cite{lin2024quantum,wang2023abelian}. Note that
$\Lm(a)^{\mathsf T}=\Lm(a^*)$ and $\Rm(b)^{\mathsf T}=\Rm(b^*)$. For Abelian $G$ these are the
bivariate-bicycle codes \cite{bravyi2024}. The qubits are $G\sqcup G$ and a vector on them is a
pair $(x,y)\in\FG^2$. The row of $H_X$ indexed by $r\in G$ is $(a^*r,\,rb^*)$ and the row of $H_Z$
indexed by $t$ is $(tb,\,at)$, so
\begin{equation}\label{eq:rowspaces}
\rowsp H_X=\{(a^*s,\,sb^*):s\in\FG\},\qquad \rowsp H_Z=\{(tb,\,at):t\in\FG\}.
\end{equation}
Let $\lambda_g$ be the permutation $x\mapsto gx$ on both blocks and $\rho_h$ the permutation
$x\mapsto xh$.

\begin{proposition}\label{prop:aut}
$\lambda_g$ is an automorphism of $\Code(a,b)$ whenever $ga=ag$ in $\FG$, and $\rho_h$ whenever
$hb=bh$. The group of such $g$ is the set-stabiliser $\Cz(a)=\{g: g\,\supp(a)\,g^{-1}=\supp(a)\}$
of the support under conjugation.
\end{proposition}

This sharpens the observation of Lin and Pryadko that the centraliser of the subgroup generated
by $\supp(a)$ acts on the code \cite{lin2024quantum}, and Eberhardt and Steffan noted that their
commutative treatment could be extended by keeping track of left and right multiplication
separately \cite{eberhardt2024fold}. If $\supp(a)$ is a union of conjugacy classes, $\Cz(a)=G$ and the whole
group acts by left translation, whatever $b$ is. When $G$ is the direct product of a non-Abelian
group with $\ZZ_m$, a central element $(e,y^j)$ together with a class of size two, such as
$\{(r,y^i),(r^2,y^i)\}$ in $S_3\times\ZZ_m$, gives $|\supp a|=3$ and check weight six. Lifted
surgery organises the logical space by the primitive central idempotents of the symmetry algebra,
and Appendix~\ref{app:comeas} shows how to find them when the symmetry group is non-Abelian.

\subsection{Gauging surgery and co-measurable sets}\label{sec:gauging}\label{sec:comeas}

A $Z$-type logical operator $Z_v$ with support $v$ is measured by gauging \cite{williamson2024}.
We choose a connected graph on the vertex set $v$, add one ancilla qubit per edge prepared in
$|+\rangle$, and add for each vertex $u\in v$ the \emph{Gauss-law} check
$G_u=Z_u\prod_{e\ni u}Z_e$, a new $Z$-type check that couples the data qubit $u$ to the ancillas
on its edges. Each $X$-check $r$ overlaps $v$ in an even number of qubits. It is deformed by $X$
on an edge set with boundary $r\cap v$ so that it commutes with the Gauss laws. Cycle checks
$\prod_{e\in c}X_e$ are then added over a basis of the cycle space. The product of all Gauss laws
is $Z_v$, so one round of the merged code measures it, and $O(d)$ rounds protect the outcome
against measurement errors. Operators with pairwise disjoint supports are measured in parallel by
attaching one graph each (Appendix~\ref{app:circuits} gives the construction we use).

A set of $Z$-logicals is co-measurable (Definition~\ref{def:gain}) if their supports are pairwise
disjoint and they are linearly independent modulo $\rowsp H_Z$. Its $m$ operators then extend to
a basis of $L_Z$ with a symplectically paired basis of $L_X$, and in that basis one merged code
measures $\bar Z_1,\dots,\bar Z_m$ simultaneously and non-destructively, while the other $k-m$
logical qubits remain as the logical qubits of the merged code. This is a parallel $Z$-basis
readout of $m$ of the $k$ logical qubits, or $m$ commuting measurements of a Pauli-based
computation compiled into this basis \cite{he2025extractors}. When $m=k$, one merged code reads
out every logical qubit.

\subsection{Notation}\label{sec:notation}

We write $D_n$ for the dihedral group acting on $\ZZ_n$ with $r:i\mapsto i+1$ and
$s_j:i\mapsto j-i$, and $S_3=D_3$ in the same notation. We also write $\ZZ_m=\langle y\rangle$ and
$\ZZ_9\rtimes_r\ZZ_{12}=\langle x,y\mid x^9=y^{12}=e,\ yxy^{-1}=x^r\rangle$, with elements written
$(x^i,y^j)$. For $r=8$ the action inverts $\langle x\rangle$ and the group has centre of order
six. $Z(G)$ denotes the centre of $G$.

\section{Automorphisms and rigidity}\label{sec:aut}

\subsection{Block-affine maps}\label{sec:blockaffine}

Translations are not the only automorphisms that the group structure provides. A \emph{block map}
is a permutation of $G$ of the form $x\mapsto g\,\varphi(x)\,h$ (sign $+$) or
$x\mapsto g\,\varphi(x^{-1})\,h$ (sign $-$) with $g,h\in G$ and $\varphi\in\Aut(G)$. These maps
form $\Hol(G)\cup\Hol(G)\iota$, the holomorph and its composition with the inversion $\iota$. A
\emph{block-affine} map of $G\sqcup G$ applies a block map $\sigma_1$ to the first block and
$\sigma_2$ to the second, either keeping the blocks, $(x,y)\mapsto(\sigma_1x,\sigma_2y)$, or
exchanging them, $(x,y)\mapsto(\sigma_2y,\sigma_1x)$. It is \emph{Tanner-level} for $\Code(a,b)$
if it permutes the $X$-check supports among themselves and the $Z$-check supports among
themselves. Every Tanner-level map preserving the code is a code automorphism, but the converse
can fail.

Three distinguished maps are $I(x,y)=(x^{-1},y^{-1})$, $W(x,y)=(y,x)$ and $J=W\circ I$. Every
block-affine map whose block maps have equal signs is $\sigma^+\circ\epsilon$ with
$\epsilon\in\{\mathrm{id},I,W,J\}$ and $\sigma^+=[g_1,h_1\,|\,g_2,h_2\,;\varphi_1,\varphi_2]:
(x,y)\mapsto(g_1\varphi_1(x)h_1,\,g_2\varphi_2(y)h_2)$. We drop $\varphi_2$ when
$\varphi_1=\varphi_2$. The parameters are determined up to
$g\varphi(x)h=(gk)(c_{k^{-1}}\circ\varphi)(x)(k^{-1}h)$, $c_k(x)=kxk^{-1}$, so the class
$[\varphi]\in\Out(G)$ is well defined. With $A=\supp a$ and $B=\supp b$, we define for
$X\subseteq G$ the \emph{twist group}
\[
T(X)=\{(h,k)\in G\times G:\ Xh=kX\},\qquad N(X)=\{h:\exists k,\ Xh=kX\},
\]
and, extending the notation of Proposition~\ref{prop:aut} from $\supp(a)$ to any $X$,
$\Cz(X)=\{g:gX=Xg\}$. We use one hypothesis on the supports alone.

\smallskip
\begin{center}
\begin{minipage}{0.9\textwidth}
(H)\ \ The one-sided set stabilisers $\{w:wA=A\}$, $\{w:Aw=A\}$, $\{w:wB=B\}$ and $\{w:Bw=B\}$
are trivial.
\end{minipage}
\end{center}
\smallskip

\noindent Under (H) the $|G|$ checks of each type have pairwise distinct supports. (H) holds for
all $22$ codes on which the classification below was tested (Section~\ref{sec:rigidity}).

\subsection{Classification}\label{sec:class}

Two lemmas supply the automorphisms of Theorem~\ref{thm:class}(a), and Lemma~\ref{lem:JIW} in
Appendix~\ref{app:twist} supplies the exchange and inversion maps. All proofs are in
Appendix~\ref{app:proofs}.

\begin{lemma}[Transport]\label{lem:transport}
For $\sigma^+=[g_1,h_1\,|\,g_2,h_2\,;\varphi]$, $\sigma^+(\Code(a,b))=\Code(a',b')$ with
$a'=g_2\,\varphi(a)\,g_1^{-1}$ and $b'=h_2^{-1}\,\varphi(b)\,h_1$. The map sends the $X$-check $r$
to the $X$-check $g_2\varphi(r)h_1$ and the $Z$-check $t$ to the $Z$-check $g_1\varphi(t)h_2$,
so it is a Tanner-level isomorphism.
\end{lemma}

\begin{lemma}[Central units]\label{lem:units}
If $u$ is a central unit of $\FG$, then $\Code(au,ub)=\Code(a,b)$.
\end{lemma}

For a central group element $z$ this only relabels the checks. For a central unit that is not a
group element, the checks of $\Code(au,ub)$ are other elements of the same row spaces.

\begin{theorem}[Block-affine automorphisms]\label{thm:class}
Let $\Code=\Code(a,b)$ and $\sigma^+=[g_1,h_1\,|\,g_2,h_2\,;\varphi_1,\varphi_2]$.
\begin{enumerate}
\item[(a)] \textup{(Exact reduction.)} If $\varphi_1=\varphi_2=\varphi$, then $\sigma^+\in\Aut(\Code)$
if and only if $\Code(a',b')=\Code(a,b)$, with $a',b'$ as in Lemma~\ref{lem:transport}. In
particular $\sigma^+\in\Aut(\Code)$ whenever $(a',b')=(au,ub)$ for a central unit $u$.
\item[(b)] \textup{(Tanner level.)} Assume (H). Then $\sigma^+$ is a Tanner-level automorphism if and
only if $\varphi_2=c_k\circ\varphi_1$ for some $k\in G$ and, with both block maps written with a
common $\varphi$, some parameter choice satisfies
\begin{equation}\label{eq:P}
\varphi(A)=g_2^{-1}\,A\,g_1,\qquad \varphi(B)=h_2\,B\,h_1^{-1}.
\end{equation}
\item[(c)] \textup{(Exchange with inversion.)} $J\circ\sigma^+\in\Aut(\Code)$ if and only if
$\sigma^+(\Code(a,b))=\Code(b^*,a^*)$; under (H) it is Tanner-level if and only if some parameter
choice satisfies $g_2\varphi(A)g_1^{-1}=B^{-1}$ and $h_2^{-1}\varphi(B)h_1=A^{-1}$.
\item[(d)] \textup{(Inversion without exchange, exchange without inversion.)} Assume (H). If a
Tanner-level automorphism has two block maps of sign $-$ without exchange, or of sign $+$ with
exchange, then $A$ and $B$ are both twisted-central: there are $\alpha,\beta\in\Aut(G)$ with
$Ak=\alpha(k)A$ and $Bk=\beta(k)B$ for all $k\in G$. If $a$ and $b$ are central, $\sigma^+\circ I$ is
a Tanner-level automorphism if and only if $g_2\varphi(A^{-1})g_1^{-1}=A$ and
$h_2^{-1}\varphi(B^{-1})h_1=B$ for some parameters, and $\sigma^+\circ W$ is one if and only if
$g_2\varphi(B)g_1^{-1}=A$ and $h_2^{-1}\varphi(A)h_1=B$.
\item[(e)] \textup{(No mixed signs.)} Assume (H) and that $G$ is non-Abelian. No Tanner-level
automorphism has block maps of different signs.
\end{enumerate}
\end{theorem}

With $\varphi=\mathrm{id}$, condition \eqref{eq:P} reads $(g_1,g_2)\in T(A)$ and
$(h_1,h_2)\in T(B)$. The pairs on the diagonal are the ordinary translations of
Proposition~\ref{prop:aut}. The others are \emph{twisted translations}, for instance
$(x,y)\mapsto(xh,\,yk)$ with $Bh=kB$ and $h\neq k$. Under (H), Theorem~\ref{thm:class}(b) also
gives the converse of Proposition~\ref{prop:aut} at the Tanner level, since $\lambda_g$ permutes
the checks only if $ga=ag$. Appendix~\ref{app:twist} describes the structure of $T(X)$ and of the
group $\Aut^{\mathrm T}_{\mathrm{ba}}(\Code)$ of Tanner-level block-affine automorphisms, which
contains the translation group $\Tr(\Code)=\{\lambda_g\rho_h:g\in\Cz(a),h\in\Cz(b)\}$
(Corollary~\ref{cor:structure}). A twist of the form $(h,whw^{-1})$ is an ordinary right
translation of an equivalent presentation (Theorem~\ref{thm:twist}(iv)), so the translation group
is not an invariant of the code.

\subsection{Rigidity}\label{sec:rigidity}

\begin{definition}[Rigidity]
$\Code(a,b)$ is \emph{rigid} if every block-affine code automorphism is a translation, that is,
$\Aba(\Code)=\Tr(\Code)$, where $\Aba(\Code)$ is the group of block-affine code automorphisms.
\end{definition}

\begin{proposition}[Rigidity criterion]\label{prop:rigid}
Assume (H) and let $G$ be non-Abelian. Then $\Aut^{\mathrm T}_{\mathrm{ba}}(\Code)=\Tr(\Code)$ iff
(R1) there are no twists, $N(A)=\Cz(a)$ and $N(B)=\Cz(b)$; (R2) every $\varphi$ satisfying
\eqref{eq:P} for some $g_1,g_2,h_1,h_2$ is inner; (R3) the system of
Theorem~\ref{thm:class}(c) has no solution; and (R4) no Tanner-level automorphism has sign $-$
without exchange or sign $+$ with exchange, which holds automatically unless $A$ and $B$ are
both twisted-central and, for central $a,b$, holds iff neither system of
Theorem~\ref{thm:class}(d) has a solution. $\Code$ is rigid iff, in addition, every block-affine
automorphism is Tanner-level.
\end{proposition}

The criterion does not settle (R4) when $A$ and $B$ are twisted-central but $a$ and $b$ are not
central, and it does not decide whether every block-affine automorphism is Tanner-level. In those
cases, and for every code in this paper, rigidity is established by exhaustive search, against
which we tested the classification on $22$ codes (Appendix~\ref{app:aut}). The block swap by an
involutive $\tau\in\Aut(G)$ with $\tau(a)=b$ known for two-block codes
\cite{eberhardt2024fold,davenport2026bicycle} is an exchange-type map in our language, and in our
set of codes it occurs only in the commutative bivariate-bicycle reference code.

\section{The non-Abelian gain}\label{sec:gain}

\begin{definition}[Co-measurability and gain]\label{def:gain}
Let $S$ be a group of code automorphisms and $v$ a $Z$-logical. A set of translates
$\{\sigma_1v,\dots,\sigma_mv\}$, $\sigma_i\in S$, is \emph{co-measurable} if the supports are
pairwise disjoint and the translates are linearly independent modulo $\rowsp H_Z$. The
co-measurable number $\mu_S(v)$ is the largest size of such a set. With $E=\Aba(\Code)$, the
\emph{non-Abelian gain} of $v$ is
\[
\gamma(v)=\frac{\mu_E(v)}{\max_{H}\mu_H(v)},
\]
the maximum being over the Abelian subgroups $H\le E$ (equivalently the maximal ones). For a
rigid code $E=\Tr(\Code)$.
\end{definition}

A naive gain uses left translations by $\Cz(a)$ in the numerator and the Abelian subgroups of
$\Tr(\Code)$ in the baseline. Both groups lie in $E$, and when $E$ is larger the baseline can be
too small. By Theorem~\ref{thm:twist}(iv) it is not even a property of the code.
Definition~\ref{def:gain} uses the same group on both sides.

\subsection{Upper bound and mirror doubling}\label{sec:mirror}

\begin{lemma}[Upper bound]\label{lem:bound}
For any group $S$ of code automorphisms and any $Z$-logical $v$, with $[Sv]\le L_Z$ the span of
the orbit modulo $\rowsp H_Z$ and $\reg(Sv)$ the union of the supports in the orbit,
\[
\mu_S(v)\;\le\;\min\Big(k,\ \dim[Sv],\ \Big\lfloor\frac{|\reg(Sv)|}{|v|}\Big\rfloor\Big)
\;\le\;\min\Big(k,\ \dim_{\F_2}\operatorname{span}(Sv),\ \Big\lfloor\frac{n}{|v|}\Big\rfloor\Big).
\]
\end{lemma}

Here $[Sv]$ is the cyclic $\F_2[S]$-submodule of $L_Z$ generated by $[v]$. A co-measurable set
of the size of the bound certifies $\mu_S$ without search. For $[[72,8,6]]$ and $[[216,16,9]]$
over $\ZZ_9\rtimes_8\ZZ_{12}$ both sides attain their bounds, so $\gamma=2$ is certified by linear
algebra and two exhibited sets alone. For a set $T$ of logical operators we write $\reg(T)$ for the union of their supports and
$[T]\le L_Z$ for their span modulo $\rowsp H_Z$.

\begin{lemma}[Mirror doubling]\label{lem:mirror}
Let $S$ act on a CSS code by automorphisms, let $T$ be a co-measurable set of $Z$-logicals and
let $s\in S$. If (i) $s\cdot\reg(T)\cap\reg(T)=\emptyset$ and (ii) $[T]\cap s[T]=0$ in $L_Z$,
then $T\cup sT$ is co-measurable. In particular, if $T\subseteq Hv$ for a subgroup $H\le S$ has maximum size and
some $s\in S$ satisfies (i) and (ii) for $T$, then $\mu_S(v)\ge2\mu_H(v)$.
\end{lemma}

Both conditions are checked by linear algebra. For $[[72,8,6]]$ over $S_3\times\ZZ_6$ the orbit
of the best Abelian subgroup of left translations covers all $72$ qubits, and of its $81$
co-measurable $4$-subsets, $45$ are mirrored in the sense of (i) and (ii) by a reflection
$s_j\in S_3$ (further examples in Appendix~\ref{app:reeval}).

\subsection{Twisted translations}\label{sec:twisted}

Twisted translations undo the mirror in two ways, through an orbit transfer that maps the logical
into the orbit of an Abelian group (Lemma~\ref{lem:transfer}) and through a twisted mirror that
doubles an Abelian set inside an Abelian group (Lemma~\ref{lem:tmirror}). Both lemmas, a code
that keeps a gain despite its twists (Remark~\ref{rem:counter}) and the evidence that rigidity is
neither necessary nor sufficient for a gain are in Appendix~\ref{app:reeval}. Rigidity is still
useful for certification: for a rigid code the block-affine baseline is the translation
baseline, which we compute exhaustively, so no block-affine automorphism can hide an Abelian group
that closes the gap.

\subsection{The coset-cover bound}\label{sec:cover}

The mirror lemma doubles an Abelian co-measurable set. How much further any number of
non-commuting elements can go is limited by a count of cosets.

\begin{proposition}[Coset cover]\label{thm:cover}
Let $E$ be a group of automorphisms of a CSS code, $v$ a $Z$-logical and $H\le E$ a subgroup, and
let $c_H(v)$ be the number of distinct sets $\sigma(Hv)$ with $\sigma\in E$. Then
$c_H(v)\le[E:H]$ and $\mu_E(v)\le c_H(v)\,\mu_H(v)$. Consequently
\[
\gamma(v)\;\le\;\min_{H}c_H(v)\;\le\;[E:H_{\max}],
\]
the minimum being over the Abelian subgroups of $E$ and $H_{\max}$ an Abelian subgroup of
largest order.
\end{proposition}

\begin{proof}
The set $\sigma(Hv)$ depends only on the coset $\sigma H$, so $c_H(v)\le[E:H]$, and the sets
cover $Ev$ because $\sigma v\in\sigma(Hv)$. A code automorphism preserves disjointness of
supports and independence modulo $\rowsp H_Z$, so $\sigma$ maps the co-measurable subsets of
$Hv$ bijectively onto those of $\sigma(Hv)$, and each of the latter has at most $\mu_H(v)$
elements. A co-measurable $T\subseteq Ev$ is the union of its intersections with the $c_H(v)$
sets, each of them co-measurable, so $|T|\le c_H(v)\mu_H(v)$. If $H^\ast$ attains the baseline,
then for every Abelian $H$ we get $\mu_E(v)\le c_H(v)\mu_H(v)\le c_H(v)\mu_{H^\ast}(v)$, hence
$\gamma(v)\le c_H(v)$. Taking $H=H_{\max}$ gives the last inequality.
\end{proof}

An Abelian subgroup of index two in $E$ therefore caps the gain at two, the value the mirror
lemma reaches. A gain of three needs a block-affine group in which every Abelian subgroup has
index at least three. Even then the bound is the number of distinct translates of an Abelian
orbit, which can be smaller than the index.

\section{Codes with certified gains}\label{sec:codes}

\subsection{Gain two}\label{sec:gain2}

Over the full block-affine group most of the gains seen against translations alone disappear:
of $386$ codes with a left-only gain above one, $77$ have twists, and in $75$ of them an Abelian
subgroup of $E$ reaches the left co-measurable number (Appendix~\ref{app:reeval}). Five codes keep
$\gamma=2$: the
rigid $[[72,8,6]]$, $[[60,12,6]]$, $[[216,16,9]]$ over $\ZZ_9\rtimes_8\ZZ_{12}$ and
$[[216,12,12]]$ codes and the non-rigid $[[216,16,9]]$ code over $\ZZ_9\rtimes_2\ZZ_{12}$.

Table~\ref{tab:family} lists every code for which the block-affine group, the gain and the
distance are all certified exactly and the gain is two. Besides these five codes, it contains the
rigid gain-two candidate of highest distance from our code search (Section~\ref{sec:search}) in
six groups, and Fig.~\ref{fig:codes} places all certified codes against the search candidates.
Rigidity, $\mu_E$, the baseline and the distance are all computed exactly
(Appendices~\ref{app:aut}, \ref{app:comeas} and~\ref{app:dist}).

\begin{table}[!t]
\centering\small
\caption{Two-block codes with exact gain $\gamma=2$ over the block-affine group $E$ and exact
distance. Merged codes are in Table~\ref{tab:family_full}, supports in
Appendix~\ref{app:supports}.}
\label{tab:family}
\setlength{\tabcolsep}{3.5pt}
\begin{tabular}{llclrrcccl}
\toprule
$G$ & $[[n,k,d]]$ & $w$ & method & $|E|$ & $|Ev|$ & $\mu$ & base ($|H|$) & $\gamma$ & orbits\\
\midrule
$S_3\times\ZZ_6$ & $[[72,8,6]]$ & 6 & I & 36 & 36 & 8 & 4 (12) & 2 & n/a\\
$D_5\times\ZZ_3$ & $[[60,12,6]]$ & 8 & I, BS & 30 & 30 & 4 & 2 (6) & 2 & n/a\\
$D_4\times\ZZ_5$ & $[[80,8,10]]$ & 8 & I, BS & 40 & 40 & 4 & 2 (20) & 2 & 1 of 12 (18)\\
$S_3\times\ZZ_7$ & $[[84,4,11]]$ & 7 & I, BS & 42 & 42 & 4 & 2 (14) & 2 & 1 of 1\\
$D_4\times\ZZ_6$ & $[[96,8,12]]$ & 8 & I, BS & 48 & 48 & 4 & 2 (24) & 2 & 1 of 12 (39)\\
$S_3\times\ZZ_9$ & $[[108,4,13]]$ & 7 & BS [I: $\ge9$] & 54 & 54 & 4 & 2 (18) & 2 & 1 of 3\\
$D_4\times\ZZ_7$ & $[[112,4,13]]$ & 8 & BS [I: $\ge7$] & 56 & 56 & 4 & 2 (28) & 2 & 1 of 1\\
$\ZZ_9\rtimes_8\ZZ_{12}$ & $[[216,16,9]]$ & 7 & I, BS & 108 & 36 & 8 & 4 (12) & 2 & n/a\\
$\ZZ_9\rtimes_2\ZZ_{12}$ & $[[216,16,9]]^\ast$ & 8 & I & 648 & 12 & 8 & 4 (36) & 2 & 1 of 1\\
$\ZZ_9\rtimes_8\ZZ_{12}$ & $[[216,12,12]]$ & 7 & BS & 108 & 18 & 6 & 3 (54) & 2 & 1 of 2\\
$\ZZ_9\rtimes_8\ZZ_{12}$ & $[[216,12,12]]$ & 7 & BS [I: $\ge6$] & 108 & 18 & 6 & 3 (54) & 2 & 1 of 2\\
\bottomrule
\end{tabular}
\par\smallskip
\begin{minipage}{\linewidth}\footnotesize All codes are rigid ($|E|=|\Tr|$) except $^\ast$. $w$: check weight. Method: BS,
exact block-split enumeration; I, integer programming (in brackets, its lower bound at the time
limit). $|Ev|$: orbit size of the minimum-weight $Z$-logical $v$. $\mu=\mu_E(v)$. Base:
$\max_H\mu_H(v)$ over the maximal Abelian subgroups $H$, with the order of a best $H$. Orbits:
orbits of minimum-weight $Z$-logicals with $\gamma=2$ among those examined (in parentheses, all
that exist; n/a, only the logical of Table~\ref{tab:reeval} examined). $[[60,12,6]]$ has $d_X=6$,
$d_Z=7$. $^\ast|E|=648=6|\Tr|$ ($216$ twisted right translations, $324$ maps with an outer
$\varphi$, no exchange), $22$ maximal Abelian subgroups. The two $[[216,12,12]]$ rows are the code
of Table~\ref{tab:reeval} and the one from the search.
\end{minipage}
\end{table}

\begin{figure}[!t]
\centering
\includegraphics[width=0.6\textwidth]{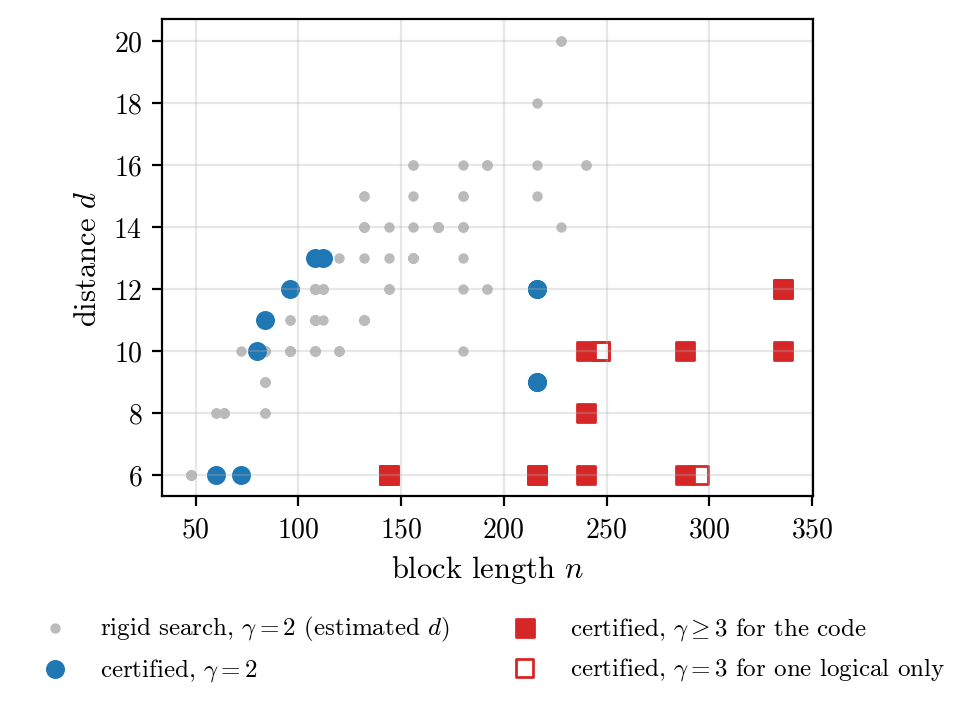}
\caption{Distance $d$ against block length $n$. Blue circles: certified gain-two codes
(Table~\ref{tab:family}). Red squares: certified gain-three codes (Table~\ref{tab:gain3}), filled
where the code-level gain is at least three and open where $\gamma=3$ holds only for one logical.
Grey dots: the $204$ rigid gain-two candidates of the search, whose distances are information-set
estimates and hence upper bounds. The certified distances are exact.}
\label{fig:codes}
\end{figure}

The gain is exactly two in every row, and no code of the search exceeds two. For $[[72,8,6]]$,
$[[84,4,11]]$, $[[108,4,13]]$ and $[[112,4,13]]$ the full gadget measures all $k$ logical qubits at
once and leaves a merged code with $k=0$, while every Abelian group reaches half of them. The
non-rigid $[[216,16,9]]$ code over $\ZZ_9\rtimes_2\ZZ_{12}$ shows that a certified gain of two does
not need rigidity. Its block-affine group is six times the translation group, yet no Abelian
subgroup of it measures more than four translates of the minimum-weight logical, against eight
for the full group.

\subsection{Gain three and more}\label{sec:gain3}

By Proposition~\ref{thm:cover}, a gain of three requires a group in which every Abelian subgroup
has index at least three. We searched the $56$ non-Abelian groups of order $21$ to $144$ with this
property, and verified every code with $\gamma\ge3$ independently (Appendix~\ref{app:search}).
There are $16$ such codes, listed in Table~\ref{tab:gain3}.

\begin{table}[!t]
\centering\small
\caption{Rigid two-block codes, all with $a$ central, with exact gain $\gamma\ge3$ over $E=\Tr$
and exact distance. Merged codes are in Table~\ref{tab:gain3_full}.}
\label{tab:gain3}
\setlength{\tabcolsep}{3.5pt}
\begin{tabular}{llclrccccccc}
\toprule
$G$ & $[[n,k,d]]$ & $w$ & method & $|E|$ & $[E:H_{\max}]$ & $c$ & orbits & $\mu$ & base & $\gamma$ & $\gamma_{\mathrm{code}}$\\
\midrule
$S_4\times\ZZ_7$ & $[[336,26,12]]$ & 8 & S & 168 & 6 & 3 & 2 of 5 & 12 & 4 & 3 & 3\\
$S_4\times\ZZ_7$ & $[[336,24,12]]$ & 8 & S & 168 & 6 & 3 & 2 of 3 & 12 & 4 & 3 & 3\\
$S_4\times\ZZ_7$ & $[[336,24,10]]$ & 8 & S & 168 & 6 & 3 & 2 of 2 & 12 & 4 & 3 & 3\\
$S_4\times\ZZ_6$ & $[[288,24,10]]$ & 8 & S & 288 & 6 & 3 & 1 of 1 & 12 & 4 & 3 & 3\\
$S_4\times\ZZ_5$ & $[[240,24,10]]$ & 8 & S & 240 & 6 & 3 & 2 of 16 & 12 & 4 & 3 & 2\\
$S_4\times\ZZ_5$ & $[[240,18,10]]$ & 8 & BS & 240 & 6 & 3 & 2 of 13 & 12 & 4 & 3 & 2\\
$\mathrm{SL}(2,3)\times\ZZ_5$ & $[[240,18,10]]$ & 8 & BS, S & 120 & 4 & 4 & 2 of 5 & 9 & 3 & 3 & 3\\
$S_4\times\ZZ_5$ & $[[240,24,8]]$ & 8 & BS, M & 120 & 6 & 3 & 1 of 1 & 12 & 4 & 3 & 3\\
$S_4\times\ZZ_5$ & $[[240,24,6]]$ & 8 & BS, M & 120 & 6 & 3 & 1 of 1 & 12 & 4 & 3 & 3\\
$\mathrm{SL}(2,3)\times\ZZ_6$ & $[[288,20,6]]$ & 7 & M & 144 & 4 & 4 & 1 of 1 & 10 & 3 & $\mathbf{10/3}$ & 10/3\\
$A_4\times\ZZ_{12}$ & $[[288,28,6]]^a$ & 7 & M & 144 & 3 & 3 & 1 of 2 & 12 & 4 & 3 & 2\\
$A_4\times\ZZ_9$ & $[[216,24,6]]$ & 7 & M & 108 & 3 & 3 & 1 of 1 & 12 & 4 & 3 & 3\\
$A_4\times\ZZ_9$ & $[[216,24,6]]$ & 7 & M & 108 & 3 & 3 & 1 of 1 & 12 & 4 & 3 & 3\\
$S_4\times\ZZ_3$ & $[[144,24,6]]$ & 8 & M & 144 & 6 & 3 & 1 of 4 & 12 & 4 & 3 & 3\\
$S_4\times\ZZ_3$ & $[[144,24,6]]$ & 8 & M & 72 & 6 & 3 & 2 of 5 & 12 & 4 & 3 & 3\\
$S_4\times\ZZ_3$ & $[[144,24,6]]$ & 8 & M & 72 & 6 & 3 & 2 of 5 & 12 & 4 & 3 & 3\\
\bottomrule
\end{tabular}
\par\smallskip
\begin{minipage}{\linewidth}\footnotesize Method: BS, block-split enumeration; M, meet in the middle; S, block split with
information-set coset scans (Appendix~\ref{app:dist}). $[E:H_{\max}]$: index of a largest Abelian
subgroup. $c$: cover bound $\min_Hc_H(v)$ (Proposition~\ref{thm:cover}). Orbits: orbits of
minimum-weight $Z$-logicals with the stated $\gamma$ among those found (the others have $\gamma$
between $1$ and $2$). $\mu$ and base refer to a best orbit.
$\gamma_{\mathrm{code}}=\max_v\mu_E(v)/\max_v\max_H\mu_H(v)$ over all orbits found. Other columns
as in Table~\ref{tab:family}. The $X$-type gain of the best orbit is two for $[[336,26,12]]$ and
$[[288,28,6]]$ and equals $\gamma$ for the others. $^a$$d_X=6$, $d_Z=8$.
\end{minipage}
\end{table}

For a best orbit of logicals, one merged code of the full gadget measures twelve logical
operators (nine or ten over $\mathrm{SL}(2,3)$) where every Abelian group of block-affine
automorphisms measures four (three), so the $[[336,24,12]]$ and $[[336,26,12]]$ codes read out
about half of their logical qubits at once. In $13$ of the $16$ codes the gain holds for the
code as a whole. Proposition~\ref{thm:cover} is tight for the $A_4$ and $S_4$ codes. Over $A_4\times\ZZ_m$ the gain three equals the index
$[E:H_{\max}]$. Over $S_4\times\ZZ_m$ the index is six, but an Abelian orbit has only three
distinct translates, and the gain is three. The $[[288,20,6]]$ code over
$\mathrm{SL}(2,3)\times\ZZ_6$ has $\gamma=10/3$, between three and its cover bound four.

\subsection{Choice of logical and search evidence}\label{sec:logical}\label{sec:search}

The gain is a property of the pair (code, logical): only $75$ of the $204$ rigid gain-two codes
of the search have $\gamma=2$ for every orbit examined, so the logical has to be chosen when the
gadget is designed. The search covered $30$ of $295$ queued groups and found rigid gain-two codes
with estimated distance up to $20$ (Fig.~\ref{fig:codes}, Appendix~\ref{app:search}), which we
treat as evidence and not as certified.

\FloatBarrier
\section{Fault tolerance}\label{sec:ft}

\subsection{Merged distance}\label{sec:merged}

A co-measurable set of size $m$ is gauged in one merged code with $O(d)$ rounds instead of
$O(md)$, and independence guarantees that the merged code has $k-m$ logical qubits and that all
$m$ outcomes are informative. Whether the merged code keeps the distance depends on the ancilla
graphs. For a graph $\Gamma$ with vertex set $V$ let $\delta U$ be the set of edges with exactly
one end in $U\subseteq V$, and let the Cheeger constant be
$h(\Gamma)=\min\{|\delta U|/|U|:\emptyset\neq U\subseteq V,\ |U|\le|V|/2\}$. The case $m=1$ of the
following lemma is Lemma~2 of Williamson and Yoder \cite{williamson2024} with $X$ and $Z$
exchanged, and the extension to several disjoint supports is in the spirit of the parallel
measurements of \cite{cowtan2025parallel} (proof in Appendix~\ref{app:proofs}).

\begin{lemma}[Merged distance]\label{lem:merged}
Let $Z_{v_1},\dots,Z_{v_m}$ be $Z$-logicals of a CSS code with pairwise disjoint supports, gauged
with connected graphs $\Gamma_i$ on $v_i$ with $h(\Gamma_i)\ge1$, with cycle checks spanning the
cycle space of every $\Gamma_i$, and with every $X$-check $r$ deformed by edge sets
$F_{r,i}\subseteq E(\Gamma_i)$ with boundary $\partial F_{r,i}=r\cap v_i$. Then the
merged code has $d'_Z\ge d_Z$ and $d'_X\ge d_X$.
\end{lemma}

Every gadget in this paper satisfies the hypotheses. Each ancilla graph is the first random
$4$-regular graph with $h\ge1$ (Appendix~\ref{app:circuits}), which for $|v|\le9$ is automatic
since $|\delta U|\ge4|U|-|U|(|U|-1)=|U|(5-|U|)\ge|U|$ whenever $|U|\le4$. Hence every merged code
in Tables~\ref{tab:family_full} and~\ref{tab:gain3_full} has $d'\ge d$. The lemma does not bound the fault
distance of the full protocol, whose time-like part relies on the $R=d$ merged rounds.

\subsection{Circuit-level simulations}\label{sec:numerics}

We simulate the full protocol (two base rounds, $R=d$ merged rounds, ancilla read-out, two base
rounds, data read-out) under circuit-level depolarising noise with a hook-aware parallel CNOT
schedule, and decode the $Z$ sector with Relay-BP \cite{muller2025relay}. A protocol fails if any
measured logical outcome is wrong (Appendix~\ref{app:circuits}). We simulate the
bivariate-bicycle $[[72,12,6]]$ code as the Abelian reference, $[[72,8,6]]$ and $[[216,16,9]]$
from Table~\ref{tab:family}, $[[240,18,10]]$ over $S_4\times\ZZ_5$ from Table~\ref{tab:gain3},
and $[[108,16,6]]$ over $S_3\times\ZZ_9$. For each code we compare the full gadget ($m$
operators, $R$ rounds) with the best Abelian gadget ($m_A=m/2$ operators, or $m/3$ for
$[[240,18,10]]$), with $m/m_A$ Abelian gadgets run in sequence, which fail with probability
$1-(1-p_A)^{m/m_A}$, and with $m$ single gadgets in sequence, which fail with probability
$1-(1-p_1)^m$ (Fig.~\ref{fig:num}). For the non-rigid $[[108,16,6]]$, Abelian groups of translations already
reach eight operators, so there the comparison is between gadget sizes.

\begin{figure}[!t]
\centering
\includegraphics[width=\textwidth]{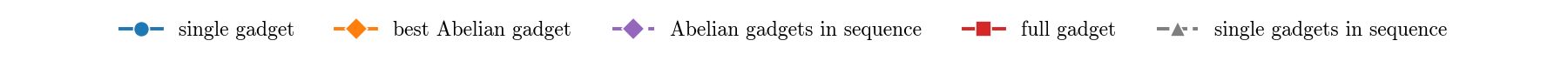}\\[2pt]
\includegraphics[width=0.325\textwidth]{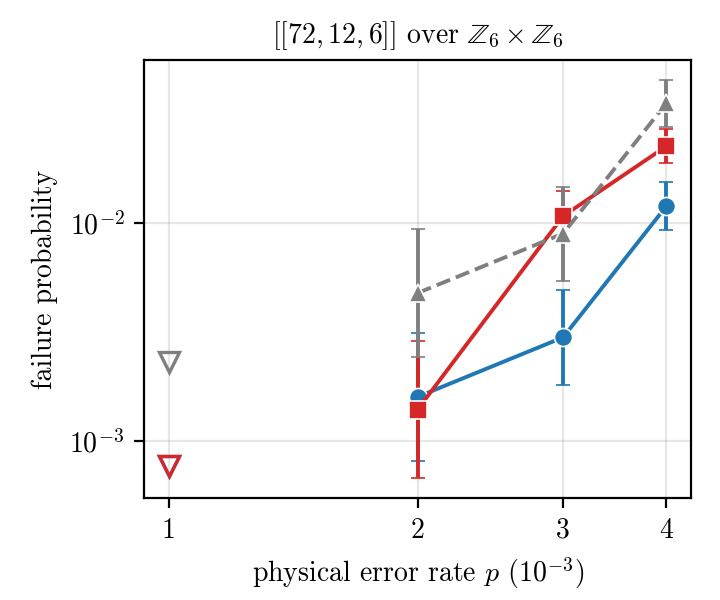}\hfill
\includegraphics[width=0.325\textwidth]{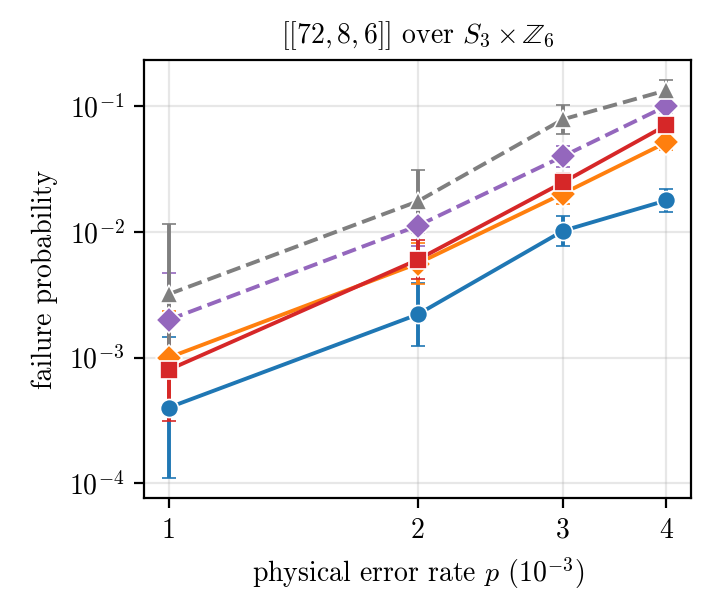}\hfill
\includegraphics[width=0.325\textwidth]{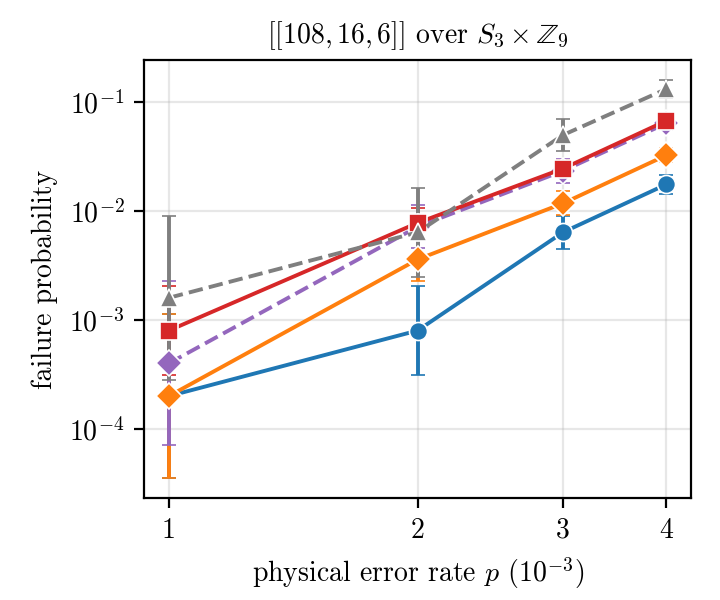}\\[2pt]
\includegraphics[width=0.325\textwidth]{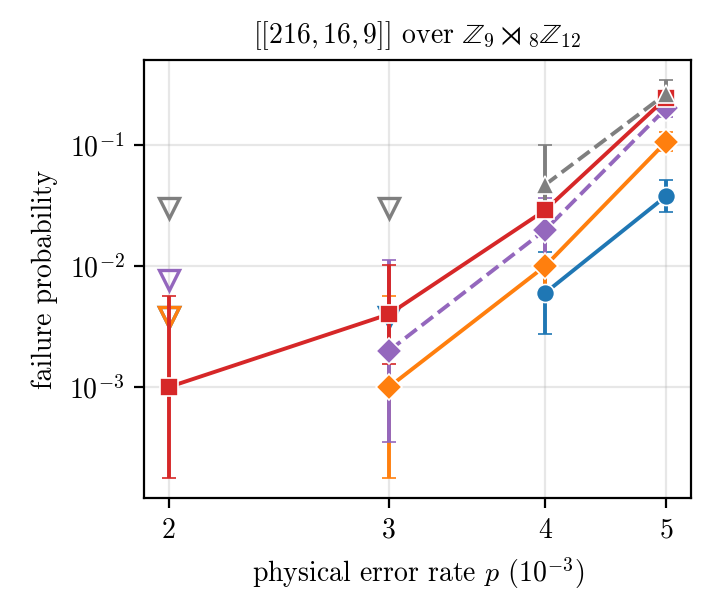}\hspace{0.02\textwidth}
\includegraphics[width=0.325\textwidth]{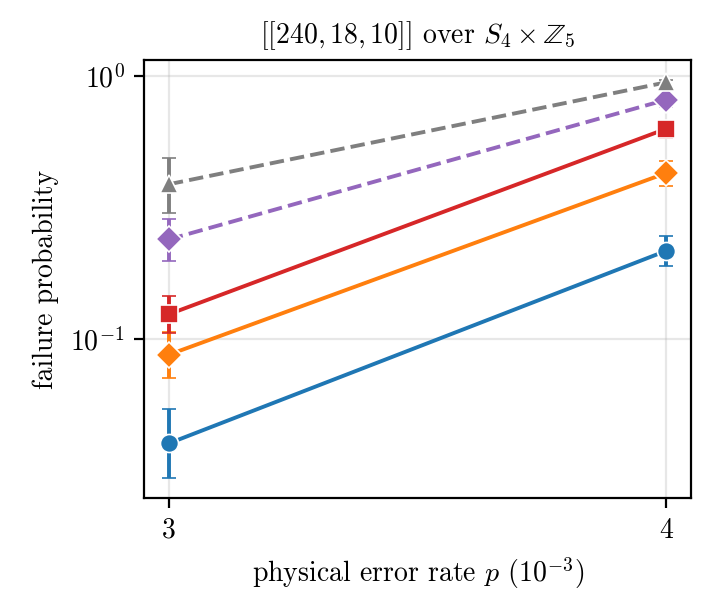}
\caption{Surgery failure probability against physical error rate, with Relay-BP decoding of the
$Z$ sector and hook-aware schedules (greedy for $[[216,16,9]]$). Each panel compares the full
gadget, which measures all its operators in $R=d$ merged rounds, with one Abelian gadget, single
gadgets and their sequences. The $[[72,12,6]]$ panel is the bivariate-bicycle code over
$\ZZ_6\times\ZZ_6$. Solid lines are measured and dashed lines derived ($1-(1-p)^j$). Error bars are $95\%$ Wilson intervals, and open
triangles mark the Wilson upper bound of points without failures.}
\label{fig:num}
\end{figure}

Whatever the error rate, the lifted gadget measures the same operators in three
($[[72,12,6]]$), eight or twelve ($[[240,18,10]]$) times fewer merged rounds than sequential
single-operator surgery, and the full gadget measures in $R$ rounds what $m/m_A$ Abelian gadgets
measure in $(m/m_A)R$. Figure~\ref{fig:ratios} shows that it does so with no loss of
reliability within statistical error. Against the Abelian gadgets run in sequence it fails $1.3$
to $1.9$ times less often on $[[72,8,6]]$ and $[[240,18,10]]$, and on $[[108,16,6]]$ and
$[[216,16,9]]$ the two protocols are indistinguishable. Against single gadgets in sequence the
point estimate favours the full gadget except at two points, both with intervals that include
one.

\begin{figure}[!t]
\centering
\includegraphics[width=0.8\textwidth]{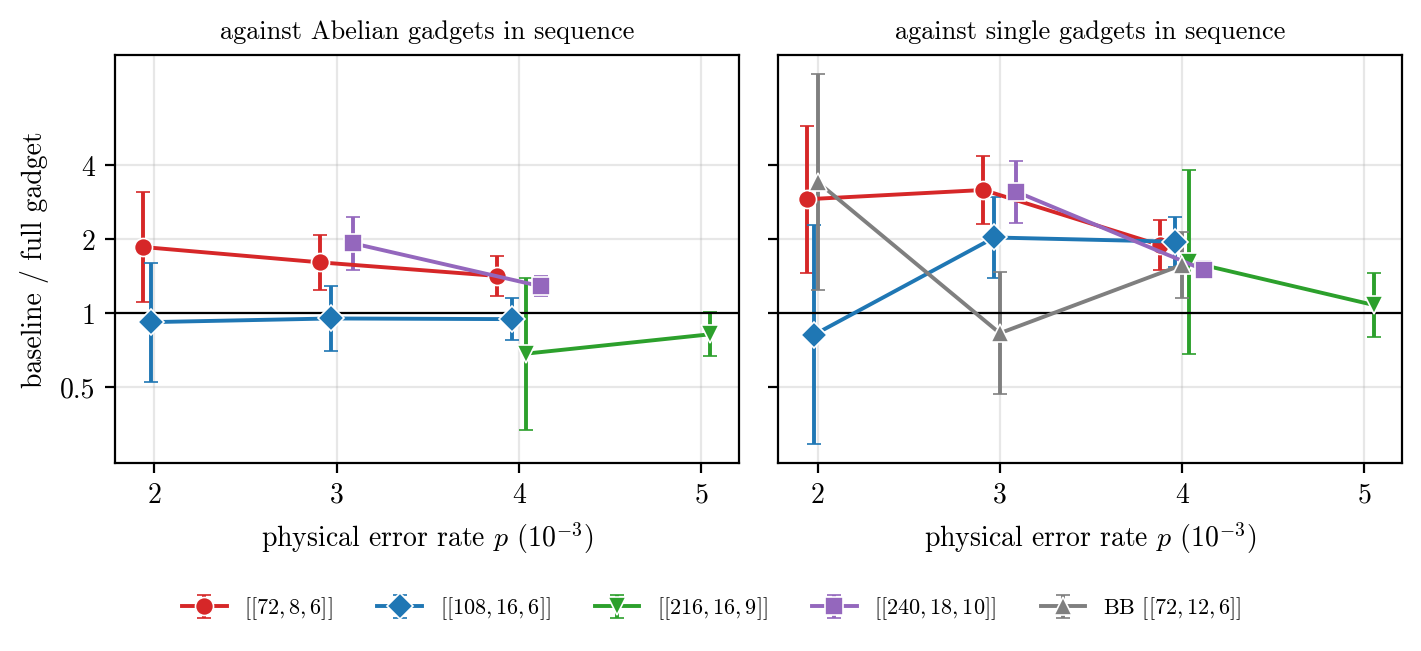}
\caption{Ratio of a baseline's failure probability to that of the full gadget, against Abelian
gadgets in sequence (left) and single gadgets in sequence (right). Error bars are $95\%$
delta-method intervals, and values above the line at one favour the full gadget. Only resolved
points are shown ($p\ge2\times10^{-3}$, and $p\ge4\times10^{-3}$ for $[[216,16,9]]$). The values
are listed in Table~\ref{tab:ratios} (Appendix~\ref{app:circuits}).}
\label{fig:ratios}
\end{figure}

The CNOT schedule matters about as much as the gadget: with a greedy edge colouring the
$[[72,8,6]]$ memory experiment fails $4.7$ times more often at $p=2\times10^{-3}$. All comparisons
are internal to one circuit builder and one decoder (Appendix~\ref{app:circuits}).

\FloatBarrier
\section{Discussion}\label{sec:discussion}

Lifted surgery carries over to non-Abelian two-block codes, but whether the non-Abelian part of
the symmetry helps depends on the baseline. Against the full block-affine group, most of the
gains seen against translations alone vanish, because twisted translations and outer maps supply
Abelian groups that the translation baseline missed. A certified gain of two remains in ten rigid
codes and one non-rigid code, and the coset-cover bound explains why two is the ceiling whenever
an Abelian subgroup has index two. In groups built from $A_4$, $S_4$ and $\mathrm{SL}(2,3)$, where
no such subgroup exists, the gain of a best logical reaches three in $15$ certified codes and
$10/3$ in one. We chose our codes for symmetry rather than rate: two-block codes over the groups of
Mitten codes, such as ours over $\ZZ_9\rtimes_8\ZZ_{12}$, carry $12$ to $16$ logical qubits where
Mitten's $[[540,108,18]]$ code over that group carries $108$ (Appendix~\ref{app:search}).

Our claims refer to the block-affine automorphism group, and since we do not compute the full
permutation automorphism group, code automorphisms of no group-theoretic origin are not excluded.
The gain belongs to a chosen logical, and three of the $16$ gain-three codes have gain two at
code level. The measured operators are fixed by the symmetry, so the gadget helps to the extent
that a computation needs them. The gadget itself does not use the symmetry: it attaches an
independent ancilla graph to each disjoint support, whereas parallel and addressable surgery
schemes measure operators whose supports overlap
\cite{cowtan2025parallel,he2025extractors,berent2026lifted}. In the simulations, all resolved
differences are at error rates close to the pseudo-thresholds of the codes, two codes have only
$1000$ shots per point, the derived sequential baselines are optimistic, and all runs are
$Z$-basis experiments (Appendix~\ref{app:circuits}). The round count, three to twelve times below
that of sequential surgery, does not depend on the decoder.

It is open whether the gain grows further with the cover bound and whether gains of three
persist at distances well above twelve. Better schedules, more shots or importance sampling at
lower error rates, the fault distance of the merged protocol and $X$-basis experiments are needed
to show that the advantage persists where the codes would be used. The largest set of disjoint,
independent minimum-weight logicals found without any symmetry, measured with the parallel
gadgets of \cite{cowtan2025parallel}, would be the natural symmetry-free baseline, which we have
not computed. An equivariant ancilla, with identical graphs and a decoder that reuses one
template, is what would make the symmetry operationally useful, and constant-time gadgets
\cite{cowtan2025fast,chang2026constant} should compose with the symmetry parallelism, although we have not constructed an equivariant
hypergraph ancilla.

\section*{Code availability}

\texttt{liftsurg} (Python; numpy, networkx, stim, scipy, relay-bp, ldpc) contains all
constructions, the automorphism searches, the exact co-measurability, enumeration and
integer-programming routines, the circuit builder, the search scripts and the result files, with
a test suite covering every construction. It is complementary to SSIP \cite{cowtan2024ssip},
which automates single merges at the check-matrix level.

\FloatBarrier
\appendix

\section{Twist group and structure}\label{app:twist}

The exchange and inversion maps act as follows.

\begin{lemma}[Exchange and inversion]\label{lem:JIW}
For every $G$, $J$ maps $\Code(a,b)$ onto $\Code(b^*,a^*)$ check to check. If $a$ and $b$ are
central in $\FG$, then $I$ maps $\Code(a,b)$ onto $\Code(a^*,b^*)$ and $W$ maps it onto
$\Code(b,a)$, again check to check.
\end{lemma}

The twist group of a support has a simple structure, which also makes the translation group
depend on the presentation of the code.

\begin{theorem}[Twist group]\label{thm:twist}
Let $X\subseteq G$.
\begin{enumerate}
\item[(i)] $T(X)$ is a subgroup of $G\times G$, and its intersection with the diagonal is
$\{(g,g):g\in\Cz(X)\}$.
\item[(ii)] If $\{w:wX=X\}$ and $\{w:Xw=X\}$ are trivial, then $T(X)=\{(h,\tau(h)):h\in N(X)\}$ for an injective
homomorphism $\tau:N(X)\to G$ whose fixed-point set is $\Cz(X)$.
\item[(iii)] Under the hypothesis of (ii), $T(X)$ is larger than its diagonal part iff
$N(X)\ne\Cz(X)$ iff there are $v_1\ne v_2$ with $Xv_1=v_2X$; this is decidable from $X$ alone
with $O(|G||X|^2)$ set operations.
\item[(iv)] If $w^{-1}Xh=z\,h\,w^{-1}X$ with $z\in Z(G)$, then $(h,\,z\,whw^{-1})\in T(X)$; in
particular $\{(h,whw^{-1}):h\in\Cz(w^{-1}X)\}\le T(X)$.
\item[(v)] If $X$ is a union of conjugacy classes, then $N(X)=\Cz(X)=G$.
\end{enumerate}
\end{theorem}

\begin{corollary}[Structure]\label{cor:structure}
Assume (H) and let $G$ be non-Abelian. The Tanner-level block-affine automorphisms form a group
$\Aut^{\mathrm T}_{\mathrm{ba}}(\Code)$. The type map to $\ZZ_2\times\ZZ_2$ (exchange, sign) is a
homomorphism, and its kernel maps onto a subgroup $O(\Code)\le\Out(G)$ by
$\sigma^+\mapsto[\varphi]$ with kernel the twist group of the code,
\[
\Tw(\Code)=\big\{[g_1,h_1\,|\,g_2,h_2]:(g_1,g_2)\in T(A),\ (h_1,h_2)\in T(B)\big\},\qquad
|\Tw(\Code)|=\frac{|T(A)|\,|T(B)|}{|Z(G)|}.
\]
Hence $|\Aut^{\mathrm T}_{\mathrm{ba}}(\Code)|=|\Tw(\Code)|\cdot|O(\Code)|\cdot\#\{\text{types
realised}\}$. The translation group
$\Tr(\Code)=\{\lambda_g\rho_h:g\in\Cz(a),h\in\Cz(b)\}$ is the diagonal part, of order
$|\Cz(a)||\Cz(b)|/|Z(G)|$.
\end{corollary}

Part (iv) of Theorem~\ref{thm:twist} says that a twist of the form $(h,whw^{-1})$ is an ordinary
right translation of the equivalent presentation $\Code(a,w^{-1}b)$, to which $\Code(a,b)$ is
carried by the translation $y\mapsto yw$ of the second block (Lemma~\ref{lem:transport}). An
Abelian baseline built from $\Cz(a)$ and $\Cz(b)$ therefore depends on the chosen presentation. By
part (v), a conjugation-invariant $a$ has no left twists, and in all our codes the twists act on
the right.

\section{Twisted translations and re-evaluated codes}\label{app:reeval}

Twisted translations can undo the mirror in two ways. Let $S=\lambda_{\Cz(a)}$ be the left action
and $\rho_{(h,k)}:(x,y)\mapsto(xh,yk)$, $(h,k)\in T(B)$, a twisted right translation.

\begin{lemma}[Orbit transfer]\label{lem:transfer}
Let $H\le E$ be Abelian with $Hv\supseteq Sv$; then $\mu_H(v)\ge\mu_S(v)$. In particular, let
$K\le\Cz(a)$ and $R\le T(B)$ be Abelian and suppose there is $\theta:R\to\Cz(a)$ with
$\rho_rv=\lambda_{\theta(r)}v$ for all $r\in R$ and $K\theta(R)=\Cz(a)$. Then
$H=\lambda_K\rho_R$ is Abelian and $Hv=Sv$.
\end{lemma}

\begin{lemma}[Twisted mirror]\label{lem:tmirror}
Let $K\le\Cz(a)$ and $R\le T(B)$ be Abelian, $T\subseteq\lambda_Kv$ co-measurable, and let
$\rho\in\rho_R$ satisfy the mirror conditions (i) and (ii) for $T$. Then
$\mu_{\lambda_K\rho_R}(v)\ge2|T|$, and likewise with the roles of $\lambda_K$ and $\rho_R$
exchanged. Consequently, if $\mu_S(v)=2\mu_{\lambda_K}(v)$ and some twist mirrors a maximum
co-measurable $T\subseteq\lambda_Kv$, the Abelian group $\lambda_K\rho_R\le E$ reaches
$\mu_S(v)$.
\end{lemma}

The hypothesis of Lemma~\ref{lem:transfer} says that $v$ is fixed by the twisted conjugation
$\lambda_{\theta(r)}^{-1}\rho_r$, a property of the logical operator rather than of the code.

\begin{remark}[A gain in the presence of twists]\label{rem:counter}
Non-rigidity does not force $\gamma=1$. Let $G=D_5\times\ZZ_3$ in the notation of Section~\ref{sec:notation} and take
$a=y+r+r^{-1}$ (central) and $b=(1+r+r^3)y^2+s_4$. This is a $[[60,8,7]]$ code with check weight
seven. Here $\Cz(b)=\ZZ_3$ but $Bs_1=s_2B$, so $T(B)$ has order $6$ and $E=\Tw(\Code)$ has order
$60$, twice the translation group, with no other types. For five different minimum-weight
logicals the left action alone gives $\mu_S=5$, while each of the six maximal Abelian subgroups
of $E$ gives at most $4$, so $\gamma\ge5/4$. Since $\mu_S$ is odd the left advantage is not a
single mirror doubling, and no twisted conjugation fixes the logical.
\end{remark}

The mirror condition of Lemma~\ref{lem:mirror} is sufficient but not necessary, and the mirroring
element need not be an involution. In $[[216,16,9]]$ over $\ZZ_9\rtimes_8\ZZ_{12}$, whose only
involution is central, the elements satisfying both conditions have orders $9$, $12$ and $18$. In
the $[[216,16,9]]$ code over $\ZZ_9\rtimes_2\ZZ_{12}$ no single element mirrors the maximal
left-Abelian set, but $\mu$ still doubles.

Table~\ref{tab:reeval} compares, for a set of codes, the naive gain $\gamma_{\mathrm{Tr}}$
measured against translations alone with the gain $\gamma$ over the full block-affine group.
Every code of Table~\ref{tab:reeval} with $\gamma_{\mathrm{Tr}}>1$ and $\gamma=1$ loses its gain
to twisted translations or outer maps. The Mitten-scale codes $[[360,16,10]]$, $[[300,16,10]]$ and
$[[336,16,10]]$, with $b$ central, already had $\gamma=1$ (the last on greedy bounds). The
mechanisms of Section~\ref{sec:twisted} account for the lost gains. A search against
translations alone found $386$ codes with a left-only gain above one. Of these, $77$ have twists
(none on the left). In $75$ of them an Abelian subgroup of $E$ reaches $\mu_S$, and in $72$ of
those every twisted right translation maps the chosen logical into its left orbit, which is the
hypothesis of Lemma~\ref{lem:transfer}. For the six codes of Table~\ref{tab:reeval} from
$[[108,20,6]]$ to $[[120,8,10]]$ that lose their gain, a best Abelian subgroup of $E$ satisfies
Lemma~\ref{lem:transfer} or Lemma~\ref{lem:tmirror}. For $[[48,10,6]]$ and $[[120,8,10]]$ orbit
transfer fails, and the gain is lost through the twisted mirror.

Rigidity is neither necessary nor sufficient for a gain. The rigid $[[84,16,7]]$ code and three
rigid codes over metacyclic groups ($(\ZZ_7\rtimes\ZZ_3)\times\ZZ_2$, $\ZZ_9\rtimes\ZZ_3$ and
$\ZZ_{13}\rtimes\ZZ_3$) have $\gamma\le1.2$. The code of Remark~\ref{rem:counter} has twists and
$\gamma\ge5/4$, the non-rigid $[[216,16,9]]$ code over $\ZZ_9\rtimes_2\ZZ_{12}$ has exact
$\gamma=2$ (Table~\ref{tab:family}), and the search of Section~\ref{sec:search} finds further
non-rigid codes with $\gamma=2$. Extra automorphisms can also raise the numerator. In $[[120,18,6]]$ over $A_4\times\ZZ_5$, with
$b$ central, automorphisms of inversion type raise $\mu_E$ to $9$, against a baseline of at least
$6$. In $[[108,16,6]]$ over $S_3\times\ZZ_9$, also with $b$ central, the translations alone
already give $10$, against a baseline of at least $8$. A second code of the search behind
Remark~\ref{rem:counter}, $[[72,12,6]]$ over $A_4\times\ZZ_3$, keeps $\mu_S=5$ against $4$ for
three of five minimum-weight logicals and has $\gamma=1$ for the other two.

\begin{table}[!t]
\centering\small
\caption{Codes re-evaluated over the full block-affine group $E=\Aba(\Code)$. $|\Tr|$, $|E|$: orders of the translation group and of $E$ (exhaustive search).
Extra automorphisms: tw, twisted translations; out, outer $\varphi$; inv, inversion without
exchange; exch, block exchange; unit, non-group central unit. $\mu_E$ and the baseline
$\max_H\mu_H$ over the Abelian subgroups $H$ of $E$ use the same minimum-weight
logical, with the order of a best $H$. $\ge$: sampled maximal Abelian subgroups,
so $\le$ marks an upper bound on the gain; $g$: greedy values. $\gamma_{\mathrm{Tr}}$: naive gain
(left action against Abelian subgroups of $\Tr$). $^\flat$Two orbits of minimum-weight logicals,
with $\gamma=2$ and $\gamma=1$; this $[[216,12,12]]$ code is distinct from the second one of
Table~\ref{tab:family}. $[[60,12,6]]$ has $d_X=6$ and $d_Z=7$.}
\label{tab:reeval}
\setlength{\tabcolsep}{3.5pt}
\resizebox{\ifdim\width>\textwidth\textwidth\else\width\fi}{!}{%
\begin{tabular}{llrrlcccccc}
\toprule
$G$ & $[[n,k,d]]$ & $|\Tr|$ & $|E|$ & extra & rigid & $\mu_E$ & $\max_H\mu_H$ ($|H|$) & $\gamma$ & $\gamma_{\mathrm{Tr}}$\\
\midrule
$\ZZ_6\times\ZZ_6$ & $[[72,12,6]]$ & 36 & 432 & out, inv, exch & & 6 & 6 (18) & 1 & 1.00\\
$S_3\times\ZZ_6$ & $[[72,8,6]]$ & 36 & 36 & none & Y & 8 & 4 (12) & \textbf{2} & 2.00\\
$S_3\times\ZZ_9$ & $[[108,16,6]]$ & 324 & 648 & inv & & 10 & $\ge8$ (54) & $\le1.25$ & 1.00\\
$S_3\times\ZZ_9$ & $[[108,20,6]]$ & 54 & 162 & tw & & 10 & 10 (54) & 1 & 1.67\\
$S_3\times\ZZ_8$ & $[[96,16,6]]$ & 48 & 192 & tw, out & & 6 & 6 (48) & 1 & 1.50\\
$S_3\times\ZZ_8$ & $[[96,8,8]]$ & 144 & 288 & out & & 4 & 4 (48) & 1 & 1.00\\
$S_3\times\ZZ_7$ & $[[84,16,7]]$ & 42 & 42 & none & Y & 6 & 5 (21) & 1.2 & 1.20\\
$S_3\times\ZZ_5$ & $[[60,8,6]]$ & 30 & 90 & tw & & 4 & 4 (30) & 1 & 2.00\\
$D_4\times\ZZ_3$ & $[[48,10,6]]$ & 24 & 48 & tw & & 4 & 4 (24) & 1 & 2.00\\
$D_4\times\ZZ_5$ & $[[80,12,6]]$ & 160 & 640 & out, inv & & 4 & 4 (40) & 1 & 1.00\\
$D_4\times\ZZ_5$ & $[[80,10,8]]$ & 80 & 80 & none & Y & 5 & 4 (40) & 1.25 & 1.00\\
$D_5\times\ZZ_3$ & $[[60,12,6]]$ & 30 & 30 & none & Y & 4 & 2 (6) & \textbf{2} & 2.00\\
$D_5\times\ZZ_4$ & $[[80,8,7]]$ & 200 & 400 & out & & 4 & 4 (40) & 1 & 1.00\\
$A_4\times\ZZ_5$ & $[[120,18,6]]$ & 720 & 5760 & out, inv, unit & & 9 & $\ge6$ (40) & $\le1.5$ & 1.00\\
$A_4\times\ZZ_3$ & $[[72,12,6]]$ & 36 & 648 & tw, out & & 4 & 4 (12) & 1 & 1.33\\
$(\ZZ_7\rtimes\ZZ_3)\times\ZZ_2$ & $[[84,14,6]]$ & 42 & 42 & none & Y & 3 & 3 (6) & 1 & 1.00\\
$\ZZ_9\rtimes\ZZ_3$ & $[[54,6,7]]$ & 27 & 27 & none & Y & 3 & 3 (9) & 1 & 1.00\\
$\ZZ_{13}\rtimes\ZZ_3$ & $[[78,6,8]]$ & 39 & 39 & none & Y & 3 & 3 (3) & 1 & 1.00\\
\midrule
$S_3\times\ZZ_{10}$ & $[[120,8,10]]$ & 60 & 120 & tw & & 4 & 4 (60) & 1 & 2.00\\
$\ZZ_9\rtimes_8\ZZ_{12}$ & $[[216,16,9]]$ & 108 & 108 & none & Y & 8 & 4 (12) & \textbf{2} & 2.00\\
$\ZZ_9\rtimes_2\ZZ_{12}$ & $[[216,16,9]]$ & 108 & 648 & tw, out & & 8 & 4 (36) & \textbf{2} & 2.00\\
$\ZZ_9\rtimes_8\ZZ_{12}$ & $[[216,12,12]]$ & 108 & 108 & none & Y & 6 & 3 (54) & \textbf{2}$^\flat$ & 2.00\\
$S_3\times\ZZ_{27}$ & $[[324,16,7{\le}d{\le}13]]$ & 162 & 486 & tw & & 8 & 8 (162) & 1 & 2.00$^g$\\
$S_3\times\ZZ_{30}$ & $[[360,32,8]]$ & 180 & 360 & out & & $\ge12^g$ & $\ge12^g$ (36) & 1$^g$ & 2.00$^g$\\
\bottomrule
\end{tabular}}
\end{table}

\section{Proofs}\label{app:proofs}

\begin{proof}[Proof of Proposition~\ref{prop:aut}]
Permuting the columns of $H_X$ by $\lambda_g$ gives $[\Lm(a)\Lm(g)^{-1}\,|\,\Rm(b)\Lm(g)^{-1}]$.
Since $\Lm(g^{-1})$ commutes with every $\Rm(\cdot)$, and with $\Lm(a)$ exactly when $ga=ag$, this
is $\Lm(g^{-1})[\Lm(a)\,|\,\Rm(b)]$, a row-equivalent matrix. The same computation applies to
$H_Z=[\Rm(b^*)\,|\,\Lm(a^*)]$ because $ga=ag$ is equivalent to $g^{-1}a^*=a^*g^{-1}$. The case of
$\rho_h$ is symmetric, and $g$ commutes with $a$ iff conjugation by $g$ permutes $\supp(a)$.
\end{proof}

\begin{proof}[Proof of Lemma~\ref{lem:transport}]
$\varphi$ extends to an algebra automorphism of $\FG$ commuting with ${}^*$. For $s\in\FG$ put
$s'=g_2\varphi(s)h_1$, a bijection of $\FG$ sending group elements to group elements. Then
$\sigma^+(a^*s,sb^*)=(g_1\varphi(a^*)g_2^{-1}s',\;s'h_1^{-1}\varphi(b^*)h_2)=(a'^*s',s'b'^*)$, and
with $t'=g_1\varphi(t)h_2$, $\sigma^+(tb,at)=(t'h_2^{-1}\varphi(b)h_1,\;g_2\varphi(a)g_1^{-1}t')
=(t'b',a't')$.
\end{proof}

\begin{proof}[Proof of Lemma~\ref{lem:JIW}]
$J(a^*s,sb^*)=(bs^*,s^*a)$ and $J(tb,at)=(t^*a^*,b^*t^*)$, which are the checks $s^*$ and $t^*$ of
$\Code(b^*,a^*)$. For central $a,b$, $I(a^*s,sb^*)=(as^*,s^*b)$, $I(tb,at)=(t^*b^*,a^*t^*)$,
$W(a^*s,sb^*)=(b^*s,sa^*)$ and $W(tb,at)=(ta,bt)$.
\end{proof}

\begin{proof}[Proof of Lemma~\ref{lem:units}]
$u^*$ is central, $((au)^*s,s(ub)^*)=(a^*(u^*s),(u^*s)b^*)$ and $(tub,aut)=((tu)b,a(tu))$. Since
$s\mapsto u^*s$ and $t\mapsto tu$ are bijections, both row spaces are unchanged.
\end{proof}

\begin{proof}[Proof of Theorem~\ref{thm:class}]
(a) is Lemma~\ref{lem:transport} with Lemma~\ref{lem:units}.

(b) Sufficiency follows from Lemma~\ref{lem:transport}. For necessity, let $\sigma^+$ carry
$X$-checks to $X$-checks and let $\pi$ be the induced map on check labels, well defined because
the checks are distinct under (H). Comparing the two blocks of the image of check $r$ gives
$g_1\varphi_1(A)^{-1}\varphi_1(r)h_1=A^{-1}\pi(r)$ and
$g_2\varphi_2(r)\varphi_2(B)^{-1}h_2=\pi(r)B^{-1}$. With $p=\pi(e)$, the first identity at $r=e$
turns the first into $A^{-1}p\,h_1^{-1}\varphi_1(r)h_1=A^{-1}\pi(r)$, and since the right
stabiliser of $A^{-1}$ is trivial, $\pi(r)=p\,h_1^{-1}\varphi_1(r)h_1$. Likewise the second gives
$\pi(r)=g_2\varphi_2(r)g_2^{-1}p$. Equating, $\varphi_2=c_k\circ\varphi_1$ with
$k=g_2^{-1}ph_1^{-1}$. Rewriting the second block map with $\varphi=\varphi_1$ and renaming
$g_2k\to g_2$, $k^{-1}h_2\to h_2$ gives $p=g_2h_1$, and the two identities at $r=e$ become
$g_2\varphi(A)g_1^{-1}=A$ and $h_2^{-1}\varphi(B)h_1=B$. Two representations with a common
automorphism differ by a central element $z$, which replaces $(a',b')$ by $(a'z,zb')$ and is
absorbed into the parameters.

(c) $J$ is an involution carrying $\Code(a,b)$ onto $\Code(b^*,a^*)$ check to check
(Lemma~\ref{lem:JIW}). Apply (a), and apply (b) to Tanner-level isomorphisms from $\Code(a,b)$
onto $\Code(b^*,a^*)$, whose proof is the same as for automorphisms, noting that (H) for
$(B^{-1},A^{-1})$ is equivalent to (H) for $(A,B)$.

(d) For sign $-$ without exchange, the first block gives
$g_1\varphi_1(r)^{-1}\varphi_1(A)h_1=A^{-1}\pi(r)$, and at $r=e$, $g_1\varphi_1(A)h_1=A^{-1}p$.
Hence every left translate $kA^{-1}$ is a right translate $A^{-1}m(k)$ with $m(k)$ unique by (H).
The map $m$ is multiplicative and injective, so $m\in\Aut(G)$, and inverting gives $Ak=m(k)A$. The
second block gives the same for $B$. Sign $+$ with exchange is $\sigma'\circ J$ with $\sigma'$ of
sign $-$. For central $a,b$ the criteria follow from Lemma~\ref{lem:JIW} as in (c).

(e) After composing with $J$, let block $1$ have sign $+$ and block $2$ sign $-$. As in (b), block
$1$ gives $\pi(r)=p\kappa_1(r)$ with $\kappa_1\in\Aut(G)$, and block $2$ gives
$Y\kappa_2(r)^{-1}=\kappa_1(r)Y$ for $Y=B^{-1}$ and some $\kappa_2\in\Aut(G)$. With
$\theta=\kappa_1\kappa_2^{-1}$, $Yk^{-1}=\theta(k)Y$ for all $k$. Then
$Y(kk')^{-1}=\theta(k')\theta(k)Y=\theta(k)\theta(k')Y$, so by (H) $\theta(G)=G$ is Abelian.
\end{proof}

\begin{proof}[Proof of Theorem~\ref{thm:twist}]
(i) If $Xh=kX$ and $Xh'=k'X$ then $Xhh'=kk'X$, and a diagonal pair lies in $T(X)$ iff $Xg=gX$.
(ii) If $(h,k),(h,k')\in T(X)$ then $kX=k'X$ and $k=k'$, so $\tau(h)=k$ is well defined on the set
$N(X)$ of first coordinates, which is a subgroup by (i). By (i), $\tau(hh')=\tau(h)\tau(h')$, so
$\tau$ is a homomorphism. If $\tau(h)=e$ then $Xh=X$, hence $h=e$ by the second hypothesis, so
$\tau$ is injective. Finally $\tau(h)=h$ iff $Xh=hX$, i.e.\ $h\in\Cz(X)$. (iii) If $v_1\in\Cz(X)$
and $Xv_1=v_2X$ then $v_1X=v_2X$ and $v_1=v_2$. For each $h$, $Xh=kX$ forces $k\in(Xh)x_0^{-1}$ for
a fixed $x_0\in X$. (iv) With $Y=w^{-1}X$, $Xh=wYh=wzhY=z(whw^{-1})X$. (v) $Xh=h(h^{-1}Xh)=hX$.
\end{proof}

\begin{proof}[Proof of Corollary~\ref{cor:structure}]
If $\varphi=c_k$ is inner, $g_ic_k(x)h_i=(g_ik)x(k^{-1}h_i)$, so the kernel of
$\sigma^+\mapsto[\varphi]$ consists of the maps with $\varphi=\mathrm{id}$, for which \eqref{eq:P}
is $Ag_1=g_2A$ and $h_2B=Bh_1$. The kernel of $T(A)\times T(B)\to\Tw(\Code)$ consists of the
parameters of the identity, $g_1=h_1^{-1}=z_1$, $g_2=h_2^{-1}=z_2$ central with $Az_1=z_2A$, so
$z_1=z_2$ by (H), of order $|Z(G)|$. The type map is well defined because
$\Hol(G)\cap\Hol(G)\iota=\emptyset$ for non-Abelian $G$ and by Theorem~\ref{thm:class}(e).
\end{proof}

\begin{proof}[Proof of Proposition~\ref{prop:rigid}]
By Corollary~\ref{cor:structure} the Tanner-level group equals $\Tr(\Code)$ iff there is one type,
$O(\Code)=1$ and $\Tw(\Code)=\Tr(\Code)$. The last holds iff $|T(A)||T(B)|=|\Cz(a)||\Cz(b)|$,
which is (R1) since $T(X)$ contains the diagonal part. The types are controlled by
Theorem~\ref{thm:class}(c) and (d).
\end{proof}

\begin{proof}[Proof of Lemma~\ref{lem:bound}]
The classes of a co-measurable set are independent in $[Sv]\le L_Z$, so $m\le\dim[Sv]\le k$, and
$\dim[Sv]=\rank[H_Z;Sv]-\rank H_Z\le\dim\operatorname{span}(Sv)$. Automorphisms preserve weight, so
the $m$ disjoint supports of size $|v|$ lie in $\reg(Sv)$ and $m|v|\le|\reg(Sv)|\le n$.
\end{proof}

\begin{proof}[Proof of Lemma~\ref{lem:mirror}]
Supports in $sT$ are pairwise disjoint because $s$ is a permutation, and they lie in
$s\cdot\reg(T)$, disjoint from $\reg(T)$ by (i). If
$\sum_ic_it_i+\sum_jd_j\,st_j\in\rowsp H_Z$, then with $w=\sum c_it_i$ and $w'=\sum d_jt_j$,
$sw'=w\in[T]\cap s[T]=0$ by (ii). Since $s$ acts invertibly on $L_Z$, $w'=0$, and independence of
$T$ gives $c=d=0$. The translates $st_j$ lie in the $S$-orbit of $v$.
\end{proof}

\begin{proof}[Proof of Lemmas~\ref{lem:transfer} and~\ref{lem:tmirror}]
The first claim of Lemma~\ref{lem:transfer} is monotonicity of $\mu$ in the orbit. Left and right
multiplications commute blockwise, so $\lambda_K$ and $\rho_R$ commute elementwise, and
$\lambda_k\rho_rv=\lambda_{k\theta(r)}v$ runs over $\lambda_{\Cz(a)}v$. Lemma~\ref{lem:tmirror}
is Lemma~\ref{lem:mirror} applied to the Abelian group $\lambda_K\rho_R$, which contains both the
group generating $T$ and the mirroring element.
\end{proof}

\begin{proof}[Proof of Lemma~\ref{lem:merged}]
For $m=1$ the lemma is Lemma~2 of Williamson and Yoder \cite{williamson2024} with $X$ and $Z$
exchanged. We give the direct argument for all $m$. Let $z=(z_D,z_E)$ be a non-trivial
$Z$-logical of the merged code, with data part $z_D$ and edge part $z_E$. Commuting with the
cycle checks of $\Gamma_i$ makes the restriction of $z_E$ to $E(\Gamma_i)$ orthogonal to its
cycle space, hence a coboundary $\delta U_i$ with $U_i\subseteq v_i$. Replacing $U_i$ by
$v_i\setminus U_i$ if necessary, $|U_i|\le|v_i|/2$. Commuting with the deformed check
$r\prod_iX(F_{r,i})$ gives
$0=|z_D\cap r|+\sum_i|\delta U_i\cap F_{r,i}|=|z_D\cap r|+\sum_i|U_i\cap r|$ modulo two,
because $|\delta U\cap F|=|U\cap\partial F|$ modulo two. So $z'=z_D+\sum_iU_i$ lies in
$\ker H_X$, and $z'$ is $z$ times the Gauss laws of the vertices in $\bigcup_iU_i$, which is
equivalent to $z$ in the merged code. If $z'$ were a stabiliser of the base code it would be one
of the merged code, and so would $z$. Hence $z'$ is a non-trivial $Z$-logical of the base code and
$|z'|\ge d_Z$. The supports $v_i$ are disjoint, so
\[
|z|=|z_D|+\sum_i|\delta U_i|\ge|z_D|+\sum_ih(\Gamma_i)|U_i|\ge|z_D|+\sum_i|U_i|\ge|z'|\ge d_Z.
\]
For the $X$ side, let $x=(x_D,x_E)$ be a non-trivial $X$-logical of the merged code. It commutes
with the base $Z$-checks, which do not touch the edges, so $x_D\in\ker H_Z$. Suppose $x_D$ were a
sum of base $X$-checks $\sum_jr_j$. Multiplying $x$ by the deformed checks $r_j\prod_iX(F_{r_j,i})$
would leave an operator on the edges alone that commutes with every Gauss law, hence has zero
boundary in each $\Gamma_i$, lies in the cycle space and is a product of cycle checks, and then
$x$ would be a stabiliser. So $x_D$ is a non-trivial $X$-logical of the base code and
$|x|\ge|x_D|\ge d_X$. No expansion is needed on this side.
\end{proof}

\section{Computing block-affine automorphisms}\label{app:aut}

Group elements are stored as multiplication tables, and $\Aut(G)$ is enumerated from a small
generating set by trying all order-preserving images of the generators. The search covers the
$|G|^2$ two-sided translations, the maps $x\mapsto g\varphi(x)$ with $\varphi\in\Aut(G)$ on both
blocks, and all pairs $(\sigma_1,\sigma_2)$ of maps in $\Hol(G)\cup\Hol(G)\iota$ applied
independently to the two blocks, with and without exchange. The last class contains every
block-affine map, and we make it exhaustive with a hash join
on row-space membership of the image of the first $X$- and $Z$-check. The image under
$(\sigma_1,\sigma_2)$ lies in the row space iff the syndromes of its two block parts against a
basis of the orthogonal complement agree, and each part depends on one map only. This test is
valid for maps that do not send checks to checks, such as central-unit relabellings. Every
remaining pair is then checked in full. A map is accepted if it carries each row space into
itself, which we test exactly against a basis of its orthogonal complement, and every accepted
map is re-verified by a rank test. The group of accepted maps is $E=\Aba(\Code)$, and the code is
rigid iff $E$ equals the group generated by $\lambda_{\Cz(a)}$ and $\rho_{\Cz(b)}$.

The $22$ codes of Section~\ref{sec:rigidity} are the $18$ codes in the upper part of
Table~\ref{tab:reeval}, the $[[120,8,10]]$ and $[[216,16,9]]$ codes over the
groups of Mitten codes, and the Mitten-scale $[[324,16]]$ and $[[360,32,8]]$ codes. The order
formula of Corollary~\ref{cor:structure} holds for the Tanner-level group of every non-Abelian
code, and no automorphism with mixed signs occurs, as Theorem~\ref{thm:class}(e) requires.
The exceptional $[[120,18,6]]$ code over $A_4\times\ZZ_5$ has $5760$ block-affine
automorphisms, of which $2880$ are Tanner-level and predicted exactly. The other $2880$ satisfy Theorem~\ref{thm:class}(a) with the non-group central unit $u=Ky^3$, where
$K$ is the class sum of the three involutions of $V_4\le A_4$. In all $22$ codes every
block-affine automorphism satisfies the unit criterion of Theorem~\ref{thm:class}(a) or its $J$,
$I$, $W$ analogue with a central unit. We have no proof that this criterion is necessary in
general.

\section{Exact co-measurability and baselines}\label{app:comeas}

Given an orbit $Sv$ we reduce each translate modulo an echelon form of $H_Z$, build the conflict
graph of overlapping supports, and run a branch-and-bound search that keeps an incremental $\F_2$
basis of the chosen representatives and prunes with
$\min(\#\text{candidates},\lfloor|\reg|/|v|\rfloor,k-\#\text{chosen})$. The set may be assumed to
contain $v$ itself. The baseline is the maximum of the same exact search over every maximal
Abelian subgroup of $E$. We enumerate these exhaustively upward from the centre, using that a
subgroup $H$ is maximal Abelian iff it equals its centraliser. Where $E$ is large, as marked in
Table~\ref{tab:reeval}, maximal Abelian subgroups are sampled and the baseline is a lower bound.
For each code we find minimum-weight $Z$-logicals by information-set search, group them into
orbits under $E$, and evaluate up to $12$ orbits.

\paragraph{Block idempotents.} Lifted surgery organises the logical space by the primitive
central idempotents of the symmetry algebra. For a non-Abelian symmetry group $S$ of even order,
$\F_2[S]$ is not semisimple and the Wedderburn decomposition is unavailable. Since $S$ acts by
automorphisms, $L_Z$ is still an $\F_2[S]$-module and splits as $\bigoplus_ie_iL_Z$ over the
primitive central idempotents $e_i$. For $[[72,8,6]]$ over $S_3\times\ZZ_6$ the logical space is
a single eight-dimensional block. The centre $Z$ of $\F_2[S]$, spanned by the class sums, is a
commutative algebra in characteristic two, so $(x+y)^2=x^2+y^2$ and $\Frob:z\mapsto z^2$ is
$\F_2$-linear. Its fixed space, the Berlekamp subalgebra \cite{berlekamp1967,lidl1997}, is exactly
the set of central idempotents, and the idempotents of a commutative ring form a Boolean algebra
whose atoms, the minimal non-zero elements, are the primitive ones. In
practice $\Frob$ is a $c\times c$ matrix over $\F_2$, where $c$ is the number of conjugacy
classes. We enumerate its fixed space and keep the minimal non-zero elements under $e\le f$ iff
$ef=e$. The decomposition of $L_Z$ is obtained by applying each idempotent to a logical basis and
taking ranks modulo $\rowsp H_Z$.

\section{Distance certification}\label{app:dist}

Upper bounds come from a randomised information-set search, which also supplies explicit
logicals of the reported weight. The methods below certify lower bounds.

\paragraph{Exact enumeration.} When $a$ is central, left translations by $G$ are automorphisms
that act regularly on each block. A $Z$-logical $(u,w)$ satisfies $M_1u+M_2w=0$ with
$[M_1\,|\,M_2]=H_X$, and if $u\neq0$ it may be translated so that the identity lies in
$\supp u$, and likewise for $w$. If $|u|+|w|\le D-1$ then $|u|\le t_1$ or $|w|\le t_2$ for any
split $t_1+t_2=D-2$, chosen to minimise the work. In the first case we enumerate all $u$ with
$e\in\supp u$ and $|u|\le t_1$, keep those with $Au$ in the image of $B$, and scan the whole coset
$w\in B^+Au+\ker B$ for $|w|\le D-1-|u|$. The second case is symmetric, and $u=0$ or $w=0$ is
handled by enumerating $\ker B$ or $\ker A$. Every vector found is tested for membership of
$\rowsp H_Z$. If no non-trivial logical of weight below $D$ appears, $d_Z\ge D$ is certified. The
same procedure with $H_Z$ gives $d_X$, and an explicit logical of weight $D$ gives equality. The
enumeration takes seconds to minutes on the codes of Table~\ref{tab:family} ($273$\,s in total
for $[[216,16,9]]$, and $43$\,s per logical type for the first $[[216,12,12]]$ code), where
integer programming took up to hours and did not close four of them. The enumeration certified
$d=13$, $13$ and $12$ where integer programming reached only $9$, $7$ and $6$ within its
$1200$\,s limit, and $d=12$ for the first $[[216,12,12]]$ code, where integer programming closed
three of $24$ sub-problems in two hours.

\paragraph{Meet in the middle.} For the larger codes of Table~\ref{tab:gain3} the block kernels
are too large to scan. With the same translation reduction, a $Z$-logical $v$ of weight at most
$D-1$ may be assumed to contain a fixed pivot qubit in the block where it is non-zero, and it
splits as $v=x+y$ with the pivot in $x$, $|x|=\lceil|v|/2\rceil$ and $|y|=\lfloor|v|/2\rfloor$. We
enumerate all such $x$ and all $y$ of weight at most $\lfloor(D-1)/2\rfloor$ and join them on
$H_Xx=H_Xy$ through a random $\F_2$-linear $64$-bit hash of the syndrome, which never misses an
equal pair. Collisions are removed by the exact test $H_Xv=0$, and every $v\neq0$ found is tested
for membership of $\rowsp H_Z$.

\paragraph{Information-set coset scans.} The block split scans cosets $w\in w_0+\ker B$ of
dimension up to $18$ for vectors of weight at most $r$. We replace each scan by the
Brouwer--Zimmermann argument. Given disjoint information sets $I_1,\dots,I_m$ of $\ker B$ and any
allocation with $\sum_i(s_i+1)>r$, a vector of weight at most $r$ has weight at most $s_i$ on some
$I_i$, and it is determined by its restriction to $I_i$. Enumerating those restrictions returns
exactly the vectors of the literal scan.

\paragraph{Integer programming.} $d_Z=\min_j\min\{|x|:H_Xx=0,\ \langle\ell_j,x\rangle=1\}$ over
$X$-logicals $\ell_j$ whose orbits under the translation group span $L_X$ modulo $\rowsp H_X$,
since a $Z$-logical anticommuting with a translate $g\ell_j$ can be translated to anticommute with
$\ell_j$ at the same weight. Parity constraints are linearised with integer slack variables and
solved with HiGHS through SciPy, either with the upper bound $D-1$ (infeasibility certifies
$d\ge D$) or to optimality. On a timeout the dual bound is the certified lower bound given in
Table~\ref{tab:family}.

\paragraph{Audits.} All three methods agree with complete enumeration of $\ker H$ on $54$ draws
of $53$ distinct small codes at $D=d,d+1,d+2$ ($162$ tests each). The block split also agrees
with exact integer programming on $44$ codes, including $[[216,16,9]]$ and $[[120,8,10]]$, and
the other two methods with integer programming on the ILP-certified codes ($144$ and $76$ test
cases) and with the block split and each other where both run. The information-set scan returned
the same vectors as the literal scan in all $324$ comparisons. Deliberately shrinking one side of a
split is detected as a gap, and at $D=13$ on $[[336,26,12]]$ the information-set scan finds the
weight-$12$ logicals, as it must.

\section{Circuits, schedules, decoding and statistics}\label{app:circuits}

\paragraph{Gauging gadgets.} For a $Z$-logical $Z_v$ with $|v|\ge6$ the ancilla graph is a random
$4$-regular graph on $v$, the first of the seeds $0,1,2,\dots$ that is connected and, for
$|v|\le16$, has exact Cheeger constant $h\ge1$ by brute force over vertex subsets. For $|v|\le5$
it is the complete graph. Each row $r$ of $H_X$ is deformed by $X$ on a union of shortest paths
whose boundary is $r\cap v$, the cycle checks run over a cycle basis, and disjoint operators are
gauged by the direct sum of their gadgets. When the ancillas are read out in the $X$ basis
with outcomes $s_e$, the cycle checks force $s=\delta\psi$ for a vertex potential $\psi$,
and the byproduct on the data is $\prod_{\psi_u=1}Z_u$, unique up to the measured operator. We
track this frame exactly in the stabiliser simulations.

\paragraph{Circuits and noise.} Stim circuits \cite{gidney2021stim} implement two opening base
rounds, $R=d$ merged rounds with ancillas prepared in $|+\rangle$, ancilla read-out in the $X$
basis, two closing base rounds and data read-out. Noise is two-qubit depolarising with
probability $p$ after every CNOT, $p$ on measurement and reset, and $p/10$ on every idle data
qubit per layer. Detectors compare each check with its previous determination across the merge
and split boundaries, and all are verified deterministic in the noiseless circuit. A protocol
fails if any measured logical outcome is wrong after decoding. The circuit builder can also track
the $k-m$ unmeasured $Z$-type logical qubits, but all runs reported here count the measured
outcomes only.

\paragraph{CNOT schedules.} CNOTs are scheduled in parallel layers. For the three weight-six codes
and for $[[240,18,10]]$ the $X$-check layers are translation invariant. The $k$-th layer applies
the CNOTs of one support term of $a$ or $b$ to every check, which is a perfect matching, so a hook
fault after $t$ CNOTs leaves the same pattern of $t$ terms on every check. The hook distance is the
least number of single-qubit $X$ faults and hook faults whose combined data error is undetected by
the $Z$ checks and flips a $Z$-logical. We estimate it by information-set search, which gives
upper bounds. Of the $360$ term orders up to reversal of a weight-six code we take the one with
the largest hook distance: at most $4$ for $[[72,8,6]]$, against at most $2$ for a greedy edge
colouring, and at most $6$ for $[[108,16,6]]$ and the bivariate-bicycle code, against at most $3$
and $4$. For $[[72,8,6]]$ a stim search finds a four-fault logical error, so its circuit distance
is at most $4$. For the weight-eight $[[240,18,10]]$ code we score $240$ random term orders of the $20160$
with $20$ information-set trials each, rescore the best ten with $100$, and check the winner with
$1200$ further trials. It has hook distance at most $7$, against at most $6$ for a greedy edge
colouring.
The $[[216,16,9]]$ circuits use a greedy proper edge colouring with random restarts, with hook
distance at most $9$, and we found no lower value. In the merged codes the base layers come first
and the edge-qubit CNOTs follow in greedy layers. Only $X$-check hooks matter, because the
experiments are in the $Z$ basis and $Z$-check hooks spread $Z$ errors, which flip no $Z$-type
observable.

\paragraph{Decoding.} Decoding uses Relay-BP \cite{muller2025relay} (implementation
\texttt{relay-bp}~0.2.2, single precision) on the $Z$ sector of the detector error model with
undecomposed hyperedges. The rows are the detectors built from $Z$-basis measurements. Each fault
column is projected onto these detectors and the observables, and columns of equal projection are
merged, with their probabilities combined by $q\mapsto q(1-p)+p(1-q)$. Every observable reads
$Z$-basis measurements only, so no fault that flips an observable is lost. The correlations
between $X$ and $Z$ components of $Y$ faults are ignored, as in any CSS decoder. The decoder runs
with memory strength $0.1$ for $80$ pre-iterations, then up to $200$ relay legs of $60$
iterations each with memory strengths drawn from $(-0.24,0.66)$, stopping after the first
converged leg. On the joint detector error model a growing fraction of shots of the larger merged
circuits never converges. For the $[[216,16,9]]$ full gadget at $p=3\times10^{-3}$ the joint
decoder failed on $8\%$ of shots, while the $Z$-sector decoder gave no failure in $250$ shots and
took about $1$\,s per shot instead of $150$\,s. An earlier version of this analysis used belief
propagation with localised-statistics decoding \cite{roffe2020,hillmann2024} (min-sum with scaling
$0.9$, parallel schedule, $30$ iterations, combination-sweep order six). Where it disagreed with Relay-BP on identical samples its correction was
typically about twice as heavy in log-likelihood weight, and it overstated failure rates by up to
about an order of magnitude. A serial schedule removed most of the difference.

\begin{table}[!t]
\centering\small
\caption{Surgery failure probabilities in percent (Relay-BP, $Z$ sector), with the $95\%$ Wilson
interval as sub- and superscript. Runs go to $5000$ shots ($1000$ for $[[216,16,9]]$ and
$[[240,18,10]]$) or stop at the first batch of $200$ shots that reaches $200$ failures ($150$), so saturated
points have fewer shots.
Rounds: merged rounds of each protocol. BB: bivariate bicycle. Memory: probability that any of the
$k$ logical qubits fails in $R$ rounds. Single, lifted $A$, lifted $E$: one gadget measuring $1$,
$m_A$ and $m$ operators ($m=3$ on BB; $m=12$, $m_A=4$ on $[[240,18,10]]$ for the simulated orbit;
$m=8$, $m_A=4$ otherwise). Derived columns, with Wilson endpoints mapped through the same
function: $1-(1-p_A)^{m/m_A}$, and $1-(1-p_1)^{m_A}$ ($A$ set) and $1-(1-p_1)^{m}$ (full set).}
\label{tab:num}
\resizebox{\ifdim\width>\textwidth\textwidth\else\width\fi}{!}{%
\begin{tabular}{llccccccc}
\toprule
code & $p$ & memory & single & lifted $A$ & $(m/m_A)\times$ lifted $A$ & lifted $E$ & seq.\ ($A$ set) & seq.\ (full set)\\
\midrule
BB $[[72,12,6]]$ & rounds & $R=6$ & 6 & n/a & n/a & 6 & n/a & 18\\
 & $1{\times}10^{-3}$ & $0_{0}^{0.077}$ & $0_{0}^{0.077}$ & n/a & n/a & $0_{0}^{0.077}$ & n/a & $0_{0}^{0.23}$\\
 & $2{\times}10^{-3}$ & $0.28_{0.17}^{0.47}$ & $0.16_{0.081}^{0.32}$ & n/a & n/a & $0.14_{0.068}^{0.29}$ & n/a & $0.48_{0.24}^{0.94}$\\
 & $3{\times}10^{-3}$ & $1.4_{1.1}^{1.8}$ & $0.3_{0.18}^{0.49}$ & n/a & n/a & $1.1_{0.83}^{1.4}$ & n/a & $0.9_{0.54}^{1.5}$\\
 & $4{\times}10^{-3}$ & $4.2_{3.7}^{4.8}$ & $1.2_{0.93}^{1.5}$ & n/a & n/a & $2.3_{1.9}^{2.7}$ & n/a & $3.6_{2.8}^{4.6}$\\
\midrule
$[[72,8,6]]$ & rounds & $R=6$ & 6 & 6 & 12 & 6 & 24 & 48\\
 & $1{\times}10^{-3}$ & $0.38_{0.24}^{0.59}$ & $0.04_{0.011}^{0.15}$ & $0.1_{0.043}^{0.23}$ & $0.2_{0.085}^{0.47}$ & $0.08_{0.031}^{0.21}$ & $0.16_{0.044}^{0.58}$ & $0.32_{0.088}^{1.2}$\\
 & $2{\times}10^{-3}$ & $1.5_{1.2}^{1.8}$ & $0.22_{0.12}^{0.39}$ & $0.56_{0.39}^{0.81}$ & $1.1_{0.77}^{1.6}$ & $0.6_{0.42}^{0.86}$ & $0.88_{0.49}^{1.6}$ & $1.7_{0.98}^{3.1}$\\
 & $3{\times}10^{-3}$ & $5_{4.3}^{5.7}$ & $1_{0.78}^{1.3}$ & $2_{1.7}^{2.4}$ & $4_{3.3}^{4.8}$ & $2.5_{2.1}^{2.9}$ & $4_{3.1}^{5.2}$ & $7.9_{6}^{10}$\\
 & $4{\times}10^{-3}$ & $11_{9.6}^{12}$ & $1.8_{1.4}^{2.2}$ & $5.1_{4.5}^{5.9}$ & $10_{8.8}^{11}$ & $7.1_{6.2}^{8}$ & $6.9_{5.7}^{8.5}$ & $13_{11}^{16}$\\
\midrule
$[[108,16,6]]$ & rounds & $R=6$ & 6 & 6 & 12 & 6 & 24 & 48\\
 & $1{\times}10^{-3}$ & $0.2_{0.11}^{0.37}$ & $0.02_{0.0035}^{0.11}$ & $0.02_{0.0035}^{0.11}$ & $0.04_{0.0071}^{0.23}$ & $0.08_{0.031}^{0.21}$ & $0.08_{0.014}^{0.45}$ & $0.16_{0.028}^{0.9}$\\
 & $2{\times}10^{-3}$ & $1.2_{0.93}^{1.5}$ & $0.08_{0.031}^{0.21}$ & $0.36_{0.23}^{0.57}$ & $0.72_{0.46}^{1.1}$ & $0.78_{0.57}^{1.1}$ & $0.32_{0.12}^{0.82}$ & $0.64_{0.25}^{1.6}$\\
 & $3{\times}10^{-3}$ & $5.1_{4.5}^{5.9}$ & $0.64_{0.45}^{0.9}$ & $1.2_{0.92}^{1.5}$ & $2.3_{1.8}^{3}$ & $2.5_{2.1}^{2.9}$ & $2.5_{1.8}^{3.6}$ & $5_{3.6}^{7}$\\
 & $4{\times}10^{-3}$ & $12_{11}^{14}$ & $1.8_{1.4}^{2.2}$ & $3.3_{2.8}^{3.8}$ & $6.4_{5.5}^{7.4}$ & $6.8_{5.9}^{7.7}$ & $6.9_{5.6}^{8.4}$ & $13_{11}^{16}$\\
\midrule
$[[216,16,9]]$ & rounds & $R=9$ & 9 & 9 & 18 & 9 & 36 & 72\\
 & $2{\times}10^{-3}$ & $0_{0}^{0.38}$ & $0_{0}^{0.38}$ & $0_{0}^{0.38}$ & $0_{0}^{0.76}$ & $0.1_{0.018}^{0.56}$ & $0_{0}^{1.5}$ & $0_{0}^{3}$\\
 & $3{\times}10^{-3}$ & $0.1_{0.018}^{0.56}$ & $0_{0}^{0.38}$ & $0.1_{0.018}^{0.56}$ & $0.2_{0.035}^{1.1}$ & $0.4_{0.16}^{1}$ & $0_{0}^{1.5}$ & $0_{0}^{3}$\\
 & $4{\times}10^{-3}$ & $0.3_{0.1}^{0.88}$ & $0.6_{0.28}^{1.3}$ & $1_{0.54}^{1.8}$ & $2_{1.1}^{3.6}$ & $2.9_{2}^{4.1}$ & $2.4_{1.1}^{5.1}$ & $4.7_{2.2}^{10}$\\
 & $5{\times}10^{-3}$ & $6_{4.7}^{7.6}$ & $3.8_{2.8}^{5.2}$ & $11_{8.9}^{13}$ & $20_{17}^{24}$ & $25_{22}^{28}$ & $14_{11}^{19}$ & $27_{20}^{35}$\\
\midrule
$[[240,18,10]]$ & rounds & $R=10$ & 10 & 10 & 30 & 10 & 40 & 120\\
 & $3{\times}10^{-3}$ & $9.7_{8}^{12}$ & $4_{3}^{5.4}$ & $8.7_{7.1}^{11}$ & $24_{20}^{29}$ & $12_{10}^{15}$ & $15_{11}^{20}$ & $39_{30}^{49}$\\
 & $4{\times}10^{-3}$ & $46_{41}^{51}$ & $22_{19}^{25}$ & $43_{38}^{48}$ & $81_{76}^{86}$ & $63_{58}^{68}$ & $62_{57}^{68}$ & $95_{92}^{97}$\\
\bottomrule
\end{tabular}}
\end{table}

\paragraph{Comparison details.} The observable for each measured logical is the product of its
Gauss laws in the first merged round combined with the final data read-out, so a flip means that
the surgery reported the wrong eigenvalue. On the bivariate-bicycle code, whose symmetry group is
Abelian, there is no separate Abelian gadget, and its $m=3$ is the exact maximum over
translations for the logical used, while Abelian groups that include its outer and inversion
automorphisms reach $\mu_E=6$ (Table~\ref{tab:reeval}). For the non-rigid $[[108,16,6]]$,
$m_A=4$ is the best Abelian group of left translations. Table~\ref{tab:num} gives the failure
probabilities and Table~\ref{tab:ratios} the ratios shown in Fig.~\ref{fig:ratios}. Against
single gadgets in sequence the two points where the point estimate favours the baseline are the
bivariate-bicycle code at $p=3\times10^{-3}$ ($1.1\%$ against $0.9\%$) and $[[108,16,6]]$ at
$p=2\times10^{-3}$. At $p=10^{-3}$ and, for $[[216,16,9]]$, up to $3\times10^{-3}$, there are too
few failures to resolve the differences. At $p=3\times10^{-3}$ the $[[240,18,10]]$ code is not in
a useful regime, as its ten-round memory fails with probability $10\%$. Another orbit of this
code admits an Abelian gadget with six operators, so two Abelian gadgets in twenty rounds could
measure twelve operators, though different ones, an alternative we have not simulated. Both
derived baselines are optimistic, since they ignore the idle errors that the other logical qubits
accumulate over the extra rounds, and the memory column of Table~\ref{tab:num} shows that these
are not small. All runs are $Z$-basis experiments, so they probe neither damage to the unmeasured
logical qubits nor the $Z$-type errors that only an $X$-basis experiment would reveal.

\paragraph{Schedules and merged codes.} With a greedy edge colouring the $[[72,8,6]]$ circuits
have circuit distance at most two, from pairs of hook faults. The absolute numbers are not those
of an optimised architecture: the ancilla graphs are $4$-regular, the merged checks have weight up
to $14$, and $[[216,16,9]]$ keeps the greedy schedule. The hook-aware schedule of $[[240,18,10]]$
reaches hook distance at most $7$, below its code distance $10$, and a check-adapted ancilla graph
for $[[216,16,9]]$ lowers its maximum merged check weight from $13$ to $10$. For the $[[360,8]]$
merged code of the $[[216,16,9]]$ code, information-set search finds a $Z$-logical of weight nine,
so $d'_Z=9$ exactly and $d'_X\ge9$ by Lemma~\ref{lem:merged}.

\begin{table}[!t]
\centering\small
\caption{Ratios plotted in Fig.~\ref{fig:ratios}: a baseline's failure probability over that of
the full gadget, with $95\%$ intervals from the delta method on the measured counts
(Table~\ref{tab:num}). Values above one
favour the full gadget. Points at $p=10^{-3}$, where the bivariate-bicycle full gadget has no
failures, and those for $[[216,16,9]]$ at $p\le3\times10^{-3}$ have too few failures to resolve
and are omitted.}
\label{tab:ratios}
\begin{tabular}{llcc}
\toprule
code & $p$ & Abelian gadgets in sequence & sequential single gadgets\\
\midrule
BB $[[72,12,6]]$ & $2{\times}10^{-3}$ & n/a & $3.4$ ($1.2$--$9.4$)\\
 & $3{\times}10^{-3}$ & n/a & $0.83$ ($0.47$--$1.5$)\\
 & $4{\times}10^{-3}$ & n/a & $1.6$ ($1.2$--$2.1$)\\
\midrule
$[[72,8,6]]$ & $2{\times}10^{-3}$ & $1.9$ ($1.1$--$3.1$) & $2.9$ ($1.5$--$5.8$)\\
 & $3{\times}10^{-3}$ & $1.6$ ($1.3$--$2.1$) & $3.2$ ($2.3$--$4.4$)\\
 & $4{\times}10^{-3}$ & $1.4$ ($1.2$--$1.7$) & $1.9$ ($1.5$--$2.4$)\\
\midrule
$[[108,16,6]]$ & $2{\times}10^{-3}$ & $0.92$ ($0.53$--$1.6$) & $0.82$ ($0.29$--$2.3$)\\
 & $3{\times}10^{-3}$ & $0.95$ ($0.70$--$1.30$) & $2.0$ ($1.4$--$3.0$)\\
 & $4{\times}10^{-3}$ & $0.95$ ($0.78$--$1.16$) & $2.0$ ($1.6$--$2.5$)\\
\midrule
$[[216,16,9]]$ & $4{\times}10^{-3}$ & $0.69$ ($0.34$--$1.40$) & $1.6$ ($0.69$--$3.8$)\\
 & $5{\times}10^{-3}$ & $0.82$ ($0.67$--$1.01$) & $1.1$ ($0.80$--$1.5$)\\
\midrule
$[[240,18,10]]$ & $3{\times}10^{-3}$ & $1.9$ ($1.5$--$2.5$) & $3.1$ ($2.3$--$4.2$)\\
 & $4{\times}10^{-3}$ & $1.3$ ($1.2$--$1.4$) & $1.5$ ($1.4$--$1.6$)\\
\bottomrule
\end{tabular}
\end{table}

\paragraph{Statistics.} Failure probabilities are binomial proportions with $95\%$ Wilson score
intervals. Derived quantities $1-(1-p)^j$ are monotone in $p$, and their intervals are the images
of the Wilson endpoints. Ratios of failure probabilities use the delta method on the measured
counts.

\section{Code search}\label{app:search}

\paragraph{Search evidence.} The six rows of Table~\ref{tab:family} that are not in
Table~\ref{tab:reeval} are the rigid gain-two candidates of highest distance in the groups we sent
to exact verification ($S_3\times\ZZ_7$, $S_3\times\ZZ_9$, $S_3\times\ZZ_{14}$, $D_4\times\ZZ_5$,
$D_4\times\ZZ_6$, $D_4\times\ZZ_7$ and $\ZZ_9\rtimes_8\ZZ_{12}$), except $S_3\times\ZZ_{14}$, whose
candidate has $\gamma=1$ at minimum weight (below). The $\ZZ_9\rtimes_8\ZZ_{12}$ candidate is a
second $[[216,12,12]]$ code, not equivalent to the first under any block-affine map. The search
covers non-Abelian groups of order $24$ to $150$ and Mitten-type semidirect products
$\ZZ_m\rtimes\ZZ_n$ up to order about $220$, with $a$ a union of conjugacy classes and $b$ any
element, both of weight three or four, and with $E$, $\mu_E$ and the baseline computed exactly. We
stopped it after $30$ of $295$ queued groups, having evaluated $25{,}465$ codes, of which
$22{,}962$ are rigid. Of these codes, $254$ have $\gamma=2$ for at least one examined
minimum-weight logical, $204$ of them rigid, and none has $\gamma>2$. Rigid gain-two codes, mostly
with $k=4$, exist over $S_3\times\ZZ_m$ for every $m$ from $5$ to $20$ except $17$, over
$D_4\times\ZZ_m$ for $m=3,\dots,7$ and over $\ZZ_9\rtimes_r\ZZ_{12}$ for $r=2,8$
(Table~\ref{tab:search}), with estimated distance growing to $20$ for $[[228,8,20]]$ over
$S_3\times\ZZ_{19}$. These distances are information-set upper bounds and the gains refer to
logicals found by a short estimate.

\paragraph{Dependence on the chosen logical.} For $[[80,8,10]]$ one of the $12$ orbits of
minimum-weight logicals examined, out of $18$, gives $\gamma=2$ and the others give $1$ to $1.5$.
For $[[96,8,12]]$ it is one of $12$ examined out of $39$. For $[[108,4,13]]$ the three orbits give
$2$, $3/2$ and $1$, and for each of the two $[[216,12,12]]$ codes the two orbits give $2$ and $1$.
Among the gain-three codes, two of the $13$ orbits of $[[240,18,10]]$ over $S_4\times\ZZ_5$ give
$\gamma=3$ and the others give $3/2$. In $[[240,24,10]]$ and $[[240,18,10]]$ over
$S_4\times\ZZ_5$ and $[[288,28,6]]$ over $A_4\times\ZZ_{12}$ another orbit lets an Abelian group
measure six operators, so their code-level gain is two.

\paragraph{Rigid search.} For each group of the queue, $a$ runs over the unions of conjugacy
classes of weight three or four up to $\Aut(G)$ and multiplication by central elements. By the
equivalence $\Code(a,b)\cong\Code(b^*,a^*)$ this also covers central $b$. The element $b$ runs
over the elements of weight three or four up to two-sided translations and the automorphisms
fixing $\supp a$, exhaustively when fewer than $40{,}000$ pairs exist and otherwise by
$15{,}000$ sampled pairs. Codes with $k\ge4$ and estimated distance at least six pass to the
second stage. This stage keeps up to $1000$ codes per group by estimated distance, replaces $b$ by
the translate with the largest centraliser, computes $E$ exhaustively, and evaluates $\mu_E$ and
the baseline exactly for up to $12$ orbits of minimum-weight logicals. All $254$ codes with
$\gamma=2$ were evaluated exactly. The other values that occur are $3/2$ ($757$ codes), $4/3$
($205$), $5/3$ ($34$), $5/4$ ($24$) and smaller, and $24{,}177$ codes have $\gamma=1$. Table~\ref{tab:search}
lists the results per group.

\begin{table}[!t]
\centering\small
\caption{Rigid gain-two codes found by the search, per group ($30$ of $295$ queued groups reached,
some partially). Codes: codes with the gain computed. Rigid $\gamma{=}2$: rigid codes with
$\gamma=2$ for some minimum-weight logical examined. All orbits: those with $\gamma=2$ for every
orbit examined. Non-rigid $\gamma{=}2$: non-rigid codes with $\gamma=2$ over $E$. Last column:
rigid gain-two code of largest estimated (information-set) distance. $^\dagger$$\gamma=2$ only for
a weight-$16$ logical, $\gamma=1$ at weight $14$. $S_3\times\ZZ_4$ and $\ZZ_9\rtimes_4\ZZ_{12}$
have $3$ and $2$ non-rigid gain-two codes, so the non-rigid column totals $50$. No rigid gain-two
code was found over $S_3\times\ZZ_4$, $S_3\times\ZZ_{17}$,
$S_3\times\ZZ_{21},\dots,S_3\times\ZZ_{25}$ and $\ZZ_9\rtimes_4\ZZ_{12}$.}
\label{tab:search}
\begin{tabular}{lrrrrl}
\toprule
$G$ & codes & rigid $\gamma{=}2$ & all orbits & non-rigid $\gamma{=}2$ & best rigid $\gamma{=}2$\\
\midrule
$S_3\times\ZZ_5$ & 299 & 3 & 1 & 0 & $[[60,4,8]]$\\
$S_3\times\ZZ_6$ & 1804 & 3 & 0 & 2 & $[[72,4,10]]$\\
$S_3\times\ZZ_7$ & 1000 & 17 & 7 & 0 & $[[84,4,11]]$\\
$S_3\times\ZZ_8$ & 1000 & 2 & 0 & 11 & $[[96,4,11]]$\\
$S_3\times\ZZ_9$ & 1000 & 44 & 19 & 0 & $[[108,4,13]]$\\
$S_3\times\ZZ_{10}$ & 1000 & 4 & 0 & 5 & $[[120,4,13]]$\\
$S_3\times\ZZ_{11}$ & 1000 & 16 & 5 & 0 & $[[132,4,15]]$\\
$S_3\times\ZZ_{12}$ & 1997 & 5 & 4 & 6 & $[[144,4,14]]$\\
$S_3\times\ZZ_{13}$ & 1000 & 11 & 7 & 0 & $[[156,4,16]]$\\
$S_3\times\ZZ_{14}$ & 999 & 10 & 7 & 3 & $[[168,16,14]]^\dagger$\\
$S_3\times\ZZ_{15}$ & 1000 & 9 & 7 & 0 & $[[180,4,16]]$\\
$S_3\times\ZZ_{16}$ & 990 & 5 & 3 & 4 & $[[192,8,16]]$\\
$S_3\times\ZZ_{18}$ & 1000 & 4 & 3 & 1 & $[[216,6,18]]$\\
$S_3\times\ZZ_{19}$ & 1000 & 3 & 2 & 0 & $[[228,8,20]]$\\
$S_3\times\ZZ_{20}$ & 1000 & 2 & 0 & 0 & $[[240,8,16]]$\\
$D_4\times\ZZ_3$ & 579 & 7 & 0 & 0 & $[[48,9,6]]$\\
$D_4\times\ZZ_4$ & 170 & 5 & 0 & 1 & $[[64,10,8]]$\\
$D_4\times\ZZ_5$ & 600 & 30 & 0 & 0 & $[[80,8,10]]$\\
$D_4\times\ZZ_6$ & 579 & 8 & 2 & 1 & $[[96,8,12]]$\\
$D_4\times\ZZ_7$ & 600 & 8 & 4 & 0 & $[[112,4,13]]$\\
$\ZZ_9\rtimes_2\ZZ_{12}$ & 600 & 4 & 1 & 9 & $[[216,4,16]]$\\
$\ZZ_9\rtimes_8\ZZ_{12}$ & 326 & 4 & 3 & 2 & $[[216,12,12]]$\\
\bottomrule
\end{tabular}
\end{table}

\paragraph{The $[[168,16,14]]$ candidate.} The $S_3\times\ZZ_{14}$ candidate is a rigid
$[[168,16,14]]$ code whose search record has $\gamma=2$, for a logical of weight $16$ taken from a
short distance estimate before the estimate dropped to $14$. At weight $14$, $20{,}000$
information-set trials find $84$ distinct $Z$-logicals, all in one translation orbit, with
$\mu=4$ equal to the baseline, so $\gamma=1$ at minimum weight. The exact enumeration shows that
no other weight-$14$ $Z$-logical has at most six qubits in either block, leaving only the $(7,7)$
split open.

\paragraph{Gain-three verification.} The $56$ groups of Section~\ref{sec:gain3} are direct
products of $A_4$, $S_4$, $\mathrm{SL}(2,3)$, the Frobenius groups of orders $20$ and $21$, the
Heisenberg group of order $27$ and metacyclic groups with cyclic groups. The two searches below
gave $16$ distinct codes with $\gamma\ge3$, one of them found by both. For each we checked
rigidity by the brute-force search over block-affine maps and computed, by exact branch and
bound, $\mu_E$ and $\mu_H$ for every maximal Abelian subgroup $H$ and every orbit of
minimum-weight logicals found by $3000$ information-set trials per type. We also computed the
distance and the Cheeger constants of the ancilla graphs exactly, and checked
Proposition~\ref{thm:cover} for every subgroup and orbit. All $16$ codes pass.

\paragraph{Rate.} In the Mitten construction with a $1\times2$ base matrix $[a_0\ a_1]$ of
weight-three entries, a central entry (a union of conjugacy classes, for instance) makes
$(a_1,a_0)$ a codeword and hence a $Z$-logical of weight six, which caps the distance at six. The
automorphism symmetry of an actual Mitten code is therefore the intersection of the centralisers
of its base entries, generically the centre of $G$. Two-block codes over the same groups reach
distance nine to twelve with a conjugation-invariant support.

\paragraph{Cover-guided search.} Over the $56$ groups of Section~\ref{sec:gain3}, with $a$
central, the first search evaluated $46{,}400$ pairs $(a,b)$ in $20{,}775$ classes. For the $93$
most promising codes we computed the block-affine group and the exact gain, and $10$ have
$\gamma\ge3$. All are rigid, lie over $S_4\times\ZZ_3$, $A_4\times\ZZ_9$, $S_4\times\ZZ_5$,
$A_4\times\ZZ_{12}$ and $\mathrm{SL}(2,3)\times\ZZ_6$, and have distances from $6$ to $10$. The
denser search covered $S_4\times\ZZ_m$ and $\mathrm{SL}(2,3)\times\ZZ_m$ for $m=5,\dots,8$ and
$A_4\times\ZZ_m$ for $m=9,\dots,15$, keeping codes with $k\ge6$ and estimated $d\ge10$. Of the
$10{,}495$ candidates we evaluated $26$ exactly, and $7$ have $\gamma=3$, with distances $10$
and $12$.

\section{Full code tables}\label{app:tables}

Tables~\ref{tab:family_full} and~\ref{tab:gain3_full} repeat Tables~\ref{tab:family}
and~\ref{tab:gain3} with the merged codes of the full gadgets. Every ancilla graph has Cheeger
constant at least one, so by Lemma~\ref{lem:merged} every merged code keeps both distances, with
check weights up to $15$.

\begin{table}[!t]
\centering\small
\caption{Full version of Table~\ref{tab:family}. Merged: merged code of the full gadget measuring
the $\mu$ operators, with its maximum check weight and the smallest Cheeger constant $h$ of its
ancilla graphs. Other columns and notes as in Table~\ref{tab:family}.}
\label{tab:family_full}
\setlength{\tabcolsep}{3.5pt}
\resizebox{\ifdim\width>\textwidth\textwidth\else\width\fi}{!}{%
\begin{tabular}{llclrrcccllcc}
\toprule
$G$ & $[[n,k,d]]$ & $w$ & method & $|E|$ & $|Ev|$ & $\mu$ & base ($|H|$) & $\gamma$ & orbits & merged & max $w$ & min $h$\\
\midrule
$S_3\times\ZZ_6$ & $[[72,8,6]]$ & 6 & I & 36 & 36 & 8 & 4 (12) & 2 & n/a & $[[168,0]]$ & 11 & 2\\
$D_5\times\ZZ_3$ & $[[60,12,6]]$ & 8 & I, BS & 30 & 30 & 4 & 2 (6) & 2 & n/a & $[[116,8]]$ & 14 & 2\\
$D_4\times\ZZ_5$ & $[[80,8,10]]$ & 8 & I, BS & 40 & 40 & 4 & 2 (20) & 2 & 1 of 12 (18) & $[[160,4]]$ & 14 & 6/5\\
$S_3\times\ZZ_7$ & $[[84,4,11]]$ & 7 & I, BS & 42 & 42 & 4 & 2 (14) & 2 & 1 of 1 & $[[172,0]]$ & 14 & 6/5\\
$D_4\times\ZZ_6$ & $[[96,8,12]]$ & 8 & I, BS & 48 & 48 & 4 & 2 (24) & 2 & 1 of 12 (39) & $[[192,4]]$ & 14 & 1\\
$S_3\times\ZZ_9$ & $[[108,4,13]]$ & 7 & BS [I: $\ge9$] & 54 & 54 & 4 & 2 (18) & 2 & 1 of 3 & $[[212,0]]$ & 13 & 4/3\\
$D_4\times\ZZ_7$ & $[[112,4,13]]$ & 8 & BS [I: $\ge7$] & 56 & 56 & 4 & 2 (28) & 2 & 1 of 1 & $[[216,0]]$ & 15 & 4/3\\
$\ZZ_9\rtimes_8\ZZ_{12}$ & $[[216,16,9]]$ & 7 & I, BS & 108 & 36 & 8 & 4 (12) & 2 & n/a & $[[360,8]]$ & 13 & 3/2\\
$\ZZ_9\rtimes_2\ZZ_{12}$ & $[[216,16,9]]^\ast$ & 8 & I & 648 & 12 & 8 & 4 (36) & 2 & 1 of 1 & $[[360,8]]$ & 12 & 3/2\\
$\ZZ_9\rtimes_8\ZZ_{12}$ & $[[216,12,12]]$ & 7 & BS & 108 & 18 & 6 & 3 (54) & 2 & 1 of 2 & $[[360,6]]$ & 10 & 1\\
$\ZZ_9\rtimes_8\ZZ_{12}$ & $[[216,12,12]]$ & 7 & BS [I: $\ge6$] & 108 & 18 & 6 & 3 (54) & 2 & 1 of 2 & $[[360,6]]$ & 10 & 1\\
\bottomrule
\end{tabular}}
\end{table}

\begin{table}[!t]
\centering\small
\caption{Full version of Table~\ref{tab:gain3}, with the merged-code columns defined as in
Table~\ref{tab:family_full}. Other columns and notes as in Table~\ref{tab:gain3}.}
\label{tab:gain3_full}
\setlength{\tabcolsep}{3.5pt}
\resizebox{\ifdim\width>\textwidth\textwidth\else\width\fi}{!}{%
\begin{tabular}{llclrccccccclcc}
\toprule
$G$ & $[[n,k,d]]$ & $w$ & method & $|E|$ & $[E:H_{\max}]$ & $c$ & orbits & $\mu$ & base & $\gamma$ & $\gamma_{\mathrm{code}}$ & merged & max $w$ & min $h$\\
\midrule
$S_4\times\ZZ_7$ & $[[336,26,12]]$ & 8 & S & 168 & 6 & 3 & 2 of 5 & 12 & 4 & 3 & 3 & $[[624,14]]$ & 15 & 1\\
$S_4\times\ZZ_7$ & $[[336,24,12]]$ & 8 & S & 168 & 6 & 3 & 2 of 3 & 12 & 4 & 3 & 3 & $[[624,12]]$ & 15 & 1\\
$S_4\times\ZZ_7$ & $[[336,24,10]]$ & 8 & S & 168 & 6 & 3 & 2 of 2 & 12 & 4 & 3 & 3 & $[[576,12]]$ & 13 & 1\\
$S_4\times\ZZ_6$ & $[[288,24,10]]$ & 8 & S & 288 & 6 & 3 & 1 of 1 & 12 & 4 & 3 & 3 & $[[528,12]]$ & 14 & 1\\
$S_4\times\ZZ_5$ & $[[240,24,10]]$ & 8 & S & 240 & 6 & 3 & 2 of 16 & 12 & 4 & 3 & 2 & $[[480,12]]$ & 14 & 1\\
$S_4\times\ZZ_5$ & $[[240,18,10]]$ & 8 & BS & 240 & 6 & 3 & 2 of 13 & 12 & 4 & 3 & 2 & $[[480,6]]$ & 14 & 1\\
$\mathrm{SL}(2,3)\times\ZZ_5$ & $[[240,18,10]]$ & 8 & BS, S & 120 & 4 & 4 & 2 of 5 & 9 & 3 & 3 & 3 & $[[420,9]]$ & 13 & 1\\
$S_4\times\ZZ_5$ & $[[240,24,8]]$ & 8 & BS, M & 120 & 6 & 3 & 1 of 1 & 12 & 4 & 3 & 3 & $[[432,12]]$ & 14 & 3/2\\
$S_4\times\ZZ_5$ & $[[240,24,6]]$ & 8 & BS, M & 120 & 6 & 3 & 1 of 1 & 12 & 4 & 3 & 3 & $[[384,12]]$ & 12 & 2\\
$\mathrm{SL}(2,3)\times\ZZ_6$ & $[[288,20,6]]$ & 7 & M & 144 & 4 & 4 & 1 of 1 & 10 & 3 & $\mathbf{10/3}$ & 10/3 & $[[408,10]]$ & 9 & 2\\
$A_4\times\ZZ_{12}$ & $[[288,28,6]]^a$ & 7 & M & 144 & 3 & 3 & 1 of 2 & 12 & 4 & 3 & 2 & $[[480,16]]$ & 9 & 3/2\\
$A_4\times\ZZ_9$ & $[[216,24,6]]$ & 7 & M & 108 & 3 & 3 & 1 of 1 & 12 & 4 & 3 & 3 & $[[360,12]]$ & 9 & 2\\
$A_4\times\ZZ_9$ & $[[216,24,6]]$ & 7 & M & 108 & 3 & 3 & 1 of 1 & 12 & 4 & 3 & 3 & $[[360,12]]$ & 9 & 2\\
$S_4\times\ZZ_3$ & $[[144,24,6]]$ & 8 & M & 144 & 6 & 3 & 1 of 4 & 12 & 4 & 3 & 3 & $[[288,12]]$ & 13 & 2\\
$S_4\times\ZZ_3$ & $[[144,24,6]]$ & 8 & M & 72 & 6 & 3 & 2 of 5 & 12 & 4 & 3 & 3 & $[[288,12]]$ & 13 & 2\\
$S_4\times\ZZ_3$ & $[[144,24,6]]$ & 8 & M & 72 & 6 & 3 & 2 of 5 & 12 & 4 & 3 & 3 & $[[288,12]]$ & 13 & 2\\
\bottomrule
\end{tabular}}
\end{table}

\FloatBarrier
\section{Supports and element indices}\label{app:supports}

With the notation of Section~\ref{sec:notation}, the codes of Table~\ref{tab:family} are:

{\small\setlength{\parindent}{0pt}
$[[72,8,6]]$ over $S_3\times\ZZ_6$: $a=\{(e,y^4),(r,e),(r^2,e)\}$, $b=\{(r,y^4),(s_1,e),(s_0,y^4)\}$.\\
$[[60,12,6]]$ over $D_5\times\ZZ_3$: $a=\{(e,e),(e,y),(r,e),(r^4,e)\}$, $b=\{(r,e),(s_2,y^2),(s_3,e),(r^4,y)\}$.\\
$[[80,8,10]]$ over $D_4\times\ZZ_5$: $a=\{(e,e),(r,y),(r^2,y^2),(r^3,y)\}$, $b=\{(e,e),(r,y),(s_0,y),(s_3,y^4)\}$.\\
$[[84,4,11]]$ over $S_3\times\ZZ_7$: $a=\{(e,e),(r,y),(r^2,y)\}$, $b=\{(e,e),(e,y^2),(r,y^3),(s_1,y^6)\}$.\\
$[[96,8,12]]$ over $D_4\times\ZZ_6$: $a=\{(e,e),(e,y),(r,e),(r^3,e)\}$, $b=\{(e,e),(r,y^2),(s_0,e),(r^2,y^4)\}$.\\
$[[108,4,13]]$ over $S_3\times\ZZ_9$: $a=\{(e,e),(r,y),(r^2,y)\}$, $b=\{(e,e),(e,y),(r,y^5),(s_1,y^2)\}$.\\
$[[112,4,13]]$ over $D_4\times\ZZ_7$: $a=\{(e,e),(e,y),(s_0,y^2),(s_2,y^2)\}$, $b=\{(e,e),(r,y^2),(s_1,y),(s_3,y^4)\}$.\\
$[[216,16,9]]$ over $\ZZ_9\rtimes_8\ZZ_{12}$: $a=\{(e,y^2),(x^4,y^6),(x^5,y^6)\}$, $b=\{(x,y^{11}),(x^4,y^2),(x^4,y^{11}),(x^7,y^{10})\}$.\\
$[[216,16,9]]$ over $\ZZ_9\rtimes_2\ZZ_{12}$: $a=\{(e,y^2),(e,y^6),(x^3,y^2),(x^6,y^2)\}$, $b=\{(x,y^{11}),(x^7,y^5),(x^8,y),(x^8,y^3)\}$.\\
$[[216,12,12]]$ (first), same group: $a=\{(e,y^2),(x^3,y^4),(x^6,y^4)\}$, $b=\{(e,y^5),(x^6,e),(x^6,y^{10}),(x^7,y^2)\}$.\\
$[[216,12,12]]$ (search), same group: $a=\{(e,e),(x^3,y^2),(x^6,y^2)\}$, $b=\{(e,e),(e,y),(e,y^3),(x,y^9)\}$.\par}

\medskip
Also mentioned in the text:

{\small\setlength{\parindent}{0pt}
$[[168,16,14]]$ over $S_3\times\ZZ_{14}$: $a=\{(e,e),(e,y),(r,y^3),(r^2,y^3)\}$, $b=\{(e,e),(e,y^3),(s_1,y^2),(s_2,y^{10})\}$.\\
$[[228,8,{\le}20]]$ over $S_3\times\ZZ_{19}$: $a=\{(e,e),(e,y),(e,y^2),(e,y^9)\}$, $b=\{(e,e),(r,y^6),(s_1,y^4),(s_2,y^{14})\}$.\par}

\medskip
Element indices for every code are in the repository.

\end{document}